\documentclass[11pt]{article}
\usepackage[margin=3cm]{geometry}

\usepackage{color}
\usepackage{tikz}

\usetikzlibrary{calc}
\usetikzlibrary{intersections}

\usepackage{amsmath}
\usepackage{amssymb}
\usepackage{amsthm}
\usepackage{mathtools}

\usepackage[labelfont=bf, skip=6pt]{caption}
\usepackage[boxed]{algorithm}
\usepackage[noend]{algpseudocode}
\usepackage{tabularx}
\usepackage{comment}
\usepackage{multirow}
\usepackage{caption}
\usepackage{subcaption}
\usepackage{booktabs}

\usepackage{xcolor}
\definecolor{darkgreen}{rgb}{0.1, 0.6, 0.3}
\definecolor{verydarkgreen}{rgb}{0,0.3,0}
\definecolor{darkblue}{rgb}{0,0,0.6}
\definecolor{lightgray}{rgb}{0.8,0.8,0.8}
\definecolor{darkred}{rgb}{0.75,0,0}
\definecolor{verydarkred}{rgb}{0.5,0,0}

\newtheorem{theorem}{Theorem}
\newtheorem{lemma}[theorem]{Lemma}
\newtheorem{proposition}[theorem]{Proposition}
\newtheorem{corollary}[theorem]{Corollary}

\theoremstyle{definition}
\newtheorem{definition}[theorem]{Definition}

\usepackage[pdfpagelabels,bookmarks=false]{hyperref}

\newcommand{\dom}{\mathrm{dom}}

 \usepackage{etoolbox}
\usepackage[pagewise]{lineno}

\newbool{journal}
\setbool{journal}{false}

\ifbool{journal}
{
\linenumbers
\def\REV#1{\textcolor{blue}{#1}}
}
{
\def\REV#1{\textcolor{black}{#1}}
}

\begin{document}

\makeatletter
\let\@fnsymbol\@arabic
\makeatother

\title{Vehicle Routing with Time-Dependent Travel Times: \\[1mm] Theory, Practice, and Benchmarks}
\author{Jannis Blauth\thanks{Research Institute for Discrete Mathematics and Hausdorff Center for Mathematics, University of Bonn.}, Stephan Held\footnotemark[1], Dirk M\"uller\footnotemark[1], Niklas Schlomberg\footnotemark[1],\\[1mm]
Vera Traub\footnotemark[1]~\thanks{Supported by the Swiss National Science Foundation grant 200021\_184622.}, Thorben Tr\"obst\thanks{University of California, Irvine.}, Jens Vygen\footnotemark[1]}

\date{}

\maketitle

 \begin{abstract}
 We develop theoretical foundations and practical algorithms for vehicle routing with time-dependent travel times.
We also provide new benchmark instances and experimental results.

First, we study basic operations on piecewise linear arrival time functions.
In particular, we devise a faster algorithm to compute the pointwise minimum of a set of piecewise linear functions
and a monotonicity-preserving variant of the Imai--Iri algorithm to approximate an arrival time function with fewer breakpoints.

Next, we show how to evaluate insertion and deletion operations in tours efficiently and update the underlying data structure faster than previously known when a tour changes.
Evaluating a tour also requires a scheduling step which is non-trivial in the presence of time windows and time-dependent travel times. We show how to perform this in linear time.

Based on these results, we develop a local search heuristic to solve real-world vehicle routing problems with various constraints efficiently and report experimental results on classical benchmarks.
Since most of these do not have time-dependent travel times, we generate and publish new benchmark instances that are based on real-world data.
This data also demonstrates the importance of considering time-dependent travel times in instances with tight time windows.

\ifbool{journal}

\end{abstract}

\ifbool{journal}
{
\subsection*{Declaration}
Declarations of interest: none

\subsection*{Acknowledgement}
We thank all students who contributed to the implementation, in particular Luise Puhlmann and Silas Rathke. We also thank our cooperation partner Greenplan, in particular Clemens Beckmann, Karin Pientka, and Jannik Silvanus.

We used map data copyrighted by OpenStreetMap contributors and available from \url{https://www.openstreetmap.org}, and
data retrieved from Uber Movement, \copyright\,2022 Uber Technologies, Inc., \url{https://movement.uber.com}.
}{}

 \section{Introduction}

In this paper we describe algorithms for vehicle routing problems with time-dependent travel times.
It is well-known that the driving times in many cities vary a lot during the day.
Nevertheless, the majority of the vehicle routing algorithms -- both academic and industrial -- still work with fixed travel times.
While this might often lead to acceptable results when there are not too many time windows, it inevitably leads to poor or infeasible results when many pickups or deliveries must take place in certain time windows.

\subsection{Problem description}

We work with a very general model, where  time-dependent travel times are modeled by piecewise linear \emph{arrival time functions}.
An arrival time function $a$ describes the arrival time $a(t)$ depending on the departure time $t$; then $a(t) - t$ is the \emph{travel time}.
Our model fulfills the monotonicity property known as the \emph{first-in-first-out principle}, which says that if we start traveling at a later time, we will not arrive earlier.
This is the most commonly used model for time-dependent travel times; see, e.g., \cite{ichoua2003vehicle}.
For a precise description of this model, see Section~\ref{sec:atf}.

 An \emph{instance} of (a simple version of) the vehicle routing problem with pickup and delivery and time windows
 consists of a list of items, each associated with a pickup address, a pickup time window, a pickup duration, a delivery address,
 a delivery time window, and a delivery duration.
 Moreover, an instance contains a road network (see Section \ref{sec:atf}) and an address where tours start and end.

 The task is to compute a set of tours such that every item appears in a tour.
 A \emph{tour} consists of a tour sequence and a feasible schedule.

A \emph{tour sequence} is a sequence $s_1,\dots, s_k$ of actions, beginning with $s_1=$ \textsc{start}, ending with  $s_k=$ \textsc{stop},
 where each action $s_i$ ($1< i < k$) is either the pickup of an item or the delivery of an item,
 such that every item that appears in an action of the tour is first picked up and then delivered.
 Each item is part of exactly two actions, the pickup at its pickup address and the delivery at its delivery address.

A \emph{schedule} of a tour assigns a start time to every action in the tour such that starting action $s_i$ at time $t_i$, performing it and then
traveling to the location of $s_{i+1}$ takes at most until $t_{i+1}$ for $1 \le i \le k-1$.
A schedule is \emph{feasible} if each action is performed in its time window.

Our model implicitly allows that vehicles can wait at  a pickup or delivery address until the time window opens.
In practice there are further constraints, for instance vehicle capacities, work time limits, and many more.
Moreover, the total cost can depend on which vehicles we use, the duration of the tours, and more.
 We will comment on how to incorporate all this when describing our overall algorithm in Section~\ref{sec:our_vrp_algorithm}.
 The simple model is sufficient to describe our main techniques.

\subsection{Outline and contributions}

We have developed fast algorithms that guarantee to obey time window constraints with respect to travel times that depend on the time of the day, while considering all major constraints that are relevant in practice.
While our overall approach is heuristic, it is based on solving certain well-defined subproblems optimally.

First, in Section~\ref{sec:atf}, we describe how we model time-dependent travel times using arrival time functions and give a thorough theoretical analysis of several basic operations which we need in our algorithm.
In particular, we devise a new algorithm to compute the pointwise minimum of $n$ arrival time functions
(in fact arbitrary piecewise linear functions)
with a total of $m$ breakpoints in time $O(m\log n)$ (Theorem~\ref{thm:atf_minimum_computation}).
Moreover, we show that a small modification of the classic Imai--Iri algorithm for polygonal line simplification \cite{II86} preserves monotonicity of piecewise linear functions (Theorem~\ref{thm:simplifying_atfs}) and can therefore be applied directly to our arrival time functions.
This shows that the post-processing step used in \cite{Batz} is obsolete.

Next, in Section~\ref{sec:upd-tours}, we discuss techniques for handling time-dependent travel times in key ingredients of local-search based algorithms for vehicle routing.
Specifically, we first describe the data structures we use to efficiently evaluate insertions (and removals) of actions (or sequences of actions) in tours.
Our data structures generalize those proposed in~\cite{Visser}, but we achieve a faster running time by storing data for fewer segments
(Theorem~\ref{thm:bst_data_structure_iterative}) and by lazy updates.
Then we give an exact algorithm for tour scheduling, i.e., for computing an optimal schedule for a given tour sequence.
While our cost model is slightly less general than the one considered in \cite{HASHIMOTO2008434}, we achieve a significantly faster running time.
In particular, our running time is linear in the number of breakpoints of the arrival time function.
We can also model penalties for arriving late in a time window in order to make schedules more robust.

In Section~\ref{sec:our_vrp_algorithm} we explain how we combine these components to obtain our overall local search algorithm.
We start by preprocessing the road network and computing a contraction hierarchy \cite{Batz, Geisberger2010engineering} to allow for efficient fastest-path queries.
Next, we construct tours using a variant of the average-regret heuristic \cite{FoisyPotvin}.
During and after constructing tours, we apply various local search operations, some of which seem to be new.
Our algorithm can handle many constraints, including heterogeneous fleets, hard and soft time windows, and various cost models.

We present experimental results in Section~\ref{sec:experimental}. Although we have designed our algorithm to
handle time-dependent travel times and many practical constraints, it yields only slightly worse results on standard
benchmarks with fixed travel times than algorithms that were designed to perform best on such instances.
We remark that our algorithm is being used successfully in various practical scenarios, in cooperation with the
startup company Greenplan\footnote{\url{www.greenplan.de}}.

Since most of the publicly available benchmarks have fixed travel times, but time-dependent travel times are necessary to evaluate our algorithm (and other algorithms in the future), we generate and publish new benchmarks based on real-world data in major cities.
Our results on these new benchmarks demonstrate once again the importance of considering time-dependent travel times in vehicle routing algorithms.
We obtain more reliable tours while avoiding higher costs that would result from working with pessimistic travel times.

\subsection{Related work}

All kinds of vehicle routing problems include the traveling salesman problem and are therefore APX-hard.
Approximation algorithms exist only for some restricted models (e.g., \cite{blauth2023improving, friggstad2021improved, Bansal, chekuri2012improved,nagarajan2011directed}),
and they are --- as of today --- not suitable for practical purposes.
There is a large body of work on mixed-integer programming models.
Recent works (e.g., \cite{Pessoa}) can solve instances with several hundred customers optimally.
For larger instances, heuristics are used in practice.
In particular, many variants of local search techniques have been proposed; see, e.g., \cite{Vidal2013cvrptw, arnold2019knowledge, AccorsiVigo:2020, QueirogaSadykovUchoa:2021,Vidal2022cvrp}.
A general overview on various approaches has been given by Toth and Vigo \cite{toth2014vehicle}.

Most works considered fixed travel times between any pair of points.
Of course, travel times in road networks vary substantially, in particular in cities and in rush hours.
To model this, time-dependent travel times were already considered in 1992 by
Malandraki and Daskin \cite{MalandrakiDaskin}, with piecewise constant travel times.
By now the most common model works with piecewise constant speeds, which leads to piecewise linear travel time functions \cite{ichoua2003vehicle}.

While some papers assume a given travel time function (or equivalently arrival
time function) for any pair of addresses \cite{donati2008time,
  figliozzi2012time, dabia2013branch, pan2021hybrid, pan2021multi},
others start with a road network where travel time functions are given for all edges of the graph \cite{Mancini,HUANG2017169, ticha2017vehicle, ben2019branch,     Gmira}. So do we.

A survey on vehicle routing with time-dependent travel times has been given by
Gendreau, Ghiani, and Guerriero \cite{gendreau2015time}.

Batz, Geisberger, Sanders, and Vetter~\cite{Batz} showed how to preprocess a road network with time-dependent travel times in order to allow relatively fast shortest path queries.
A less memory-consuming variant was proposed by Strasser, Wagner, and Zeitz \cite{Strasser}. While a contraction hierarchy needs to be recomputed only when the road network changes,
Geisberger and Sanders \cite{Geisberger2010engineering} proposed an additional instance-specific preprocessing to speed up fastest-path queries.

Computing the arrival time function from one address to the other based on the data from the road network requires certain operations on arrival time functions, in particular composition and pointwise minimum.
The same operations are also needed when composing a tour or making local changes.
While the number of breakpoints grows linearly when composing arrival time functions, this is not true for the piecewise minimum \cite{Wiernik86,HS86}.
The resulting large number of breakpoints suggests simplifying arrival time functions.
Imai and Iri showed how to approximate a continuous piecewise linear function optimally with a guaranteed error bound in linear time \cite{II86}.

Any local search based algorithm needs operations to insert items into or remove items from a tour.
Even more frequently the algorithm will evaluate the effect of such operations without actually performing them.
Even the efficient evaluation whether the insertion of an item would lead to violations of time windows is non-trivial
and requires the maintenance of data structures for fast updates.
Some papers, including \cite{Gmira}, precomputed the latest feasible departure time at any address.
If we have a fixed start time or also precompute the earliest possible arrival time at any address,
we can check feasibility with a constant number of function evaluations when inserting a sequence of constant length.
However, this does not allow to bound the tour duration (unless the start time is fixed)
or to compute the cost of the tour resulting from the insertion accurately. \cite{Gmira} used a heuristic.

The feasibility of an insert operation with respect to all time constraints
can easily be checked with $O(1)$ operations on arrival time functions
if we have precomputed the composed arrival time function
for each \emph{tour segment} (consecutive sequence of actions in a tour) \cite{fleischmann2004time,HASHIMOTO2008434}.
Based on this preprocessing, \cite{HASHIMOTO2008434} showed how to find the cheapest way
to insert an action into a tour in linear time
and how to implement certain local search operations efficiently.
However, a tour with $n$ actions has $\Theta(n^2)$ tour segments, and the update after an actual insertion in the middle
of the tour would require the recomputation of half of these tour segments.
Visser and Spliet \cite{Visser} proposed to store only $\Theta(n\log n)$ segments such that the information for
every segment can still be computed from these with a constant number of compositions.
We generalize this approach (and propose a faster update mechanism) and therefore describe it in more detail in Section~\ref{subsec:cheapest_insert}.

The evaluation of the cost of a tour (after insertion or deletion of an item) also requires a scheduling step,
in particular determining the optimal start time. 
Hashimoto, Yagiura and Ibaraki \cite{HASHIMOTO2008434} showed how to do this in quadratic time in the number of breakpoints
for a very general model, even relaxing the FIFO property of arrival time functions.

Several methods have been proposed for optimizing a single tour, including both the tour sequence and the schedule, under time window constraints and with time-dependent travel times.
Branch-and-bound algorithms have been proposed in \cite{arigliano2019time, adamo2020enhanced},
and branch-and-price algorithms have been developed in \cite{montero2017integer, lera2018integer, hansknecht2021dynamic}.
Vu, Hewitt, Boland, and Savelsbergh \cite{vu2018solving} gave a mixed-integer linear programming formulation
on a time-expanded network, which they then use within a dynamic discretization framework.
Malandraki and Dial \cite{malandraki1996restricted} described a dynamic programming algorithm for the time-dependent TSP
without time windows (working with piecewise constant travel times).
A dynamic programming algorithm for the general setting was developed by
Lera-Romero, Miranda-Bront, and Soulignac \cite{lera2020dynamic},
building on earlier work by Tilk and Irnich \cite{tilk2017dynamic}, who considered the TSP with time windows and constant travel times.

Many benchmarks for different kinds of vehicle routing problems have been published. However, most benchmarks consider constant travel times (e.g., \cite{uchoa}, \cite{arnold2019efficiently}, \cite{solomon}, \cite{homberger}). The dataset for capacitated vehicle routing by Uchoa et al.~\cite{uchoa} contains instances with up to 1000 customers. Even more recently, Arnold, Gendreau and Sörensen~\cite{arnold2019efficiently} provided very large-scale instances with up to 30,000 customers based on real-world data in Belgium.
Benchmark instances for vehicle routing with time windows and 100 customers were published by Solomon~\cite{solomon} and with up to 1000 customers by Gehring and Homberger~\cite{homberger}.
These are artificial instances based on  Euclidean distances and scattered time windows.

Due to the lack of benchmarks with time-dependent travel times, the Solomon
benchmarks have been modified by adding speed profiles to the arcs \cite{dabia2013branch}, \cite{figliozzi2012time}, \cite{donati2008time},
\cite{ichoua2003vehicle}. This resulted in instances with time-dependent
point-to-point driving times. The Euclidean distances are scaled  using speed distributions.
In \cite{figliozzi2012time} and \cite{ichoua2003vehicle} all arcs use the same
speed distribution, while \cite{dabia2013branch} and \cite{donati2008time} assign to each arc one out of a few different speed distributions.
Using average constant speeds for the planning,  \cite{donati2008time} and \cite{ichoua2003vehicle} report time window violations and a significant cost increase when evaluating such  solutions in the time-dependent model.
Similar results were observed on a real-world instance in
\cite{donati2008time}.

Gmira et al.~\cite {Gmira} and Ticha et al.~\cite{benticha} added time-dependent travel times to the
so-called NEWLET benchmarks \cite{ticha2017empirical}
for the vehicle routing problem with time windows and road network
information.
The road graphs are randomly generated sparse planar graphs with up to 200
vertices and randomly assigned speed profiles. The instances have up to 50
deliveries and are small enough to be tractable by a branch-cut-and-price
framework. Gaps but neither (optimum) solutions, solution values,  nor lower bounds have been reported.

 \section{Arrival time functions}\label{sec:atf}

 We are given a road network with travel time information that we transform into a directed graph $R$ with an arrival time function for every arc of
$R$.
 For a vehicle traveling through the arc $(v, w)$, the arrival time function $a$ describes the arrival time $a(t)$ at $w$ depending on the departure time $t$ at $v$.

 \begin{definition}\label{def:atf}
  An \emph{arrival time function (ATF)} is a piecewise linear function $a
: (-\infty, t_{\max}] \to \mathbb{R}$
  for some $t_{\max}\in\mathbb{R}$ such that
  \begin{enumerate}
   \item $a$ is continuous,\label{def:atf:c1}
   \item $a$ is non-decreasing,\label{def:atf:c2}
   \item $a(t) \ge t$ for all $t \le t_{\max}$, and\label{def:atf:c3}
   \item $a$ is initially constant, i.e., $a(t) = a(t_{\min})$ for some
$t_{\min}\le t_{\max}$ and all $t\le t_{\min}$.\label{def:atf:c4}
  \end{enumerate}
  We will also define $a(t) := \infty$ for $t> t_{\max}$. We write $\dom(a):=(-\infty, t_{\max}]$.
 \end{definition}

 This is the standard model to describe time-dependent travel times \cite{ichoua2003vehicle}.
 Each condition in Definition~\ref{def:atf} can be seen to hold for reasonable models of arrival times in road networks.
 For condition~\ref{def:atf:c4} we can simply set $t_{\min}$ to the beginning of the considered time period.
 Condition~\ref{def:atf:c3} says that travel times are always non-negative,
 and condition~\ref{def:atf:c2} (also called the \emph{first-in-first-out
principle}) simply follows from the fact that we could always wait in place before traversing an arc.
 Condition~\ref{def:atf:c1} might be the most contentious as there are some discontinuous events associated with road networks
 (e.g.\ the lowering of a bridge, the arrival of a train, an accident, etc.).
 However, in practice, arrival times are obtained via statistical sampling, and thus they rather model the \emph{expected} arrival times in our road network.
 These should always be continuous.

 On the other hand, the assumption that arrival time functions are piecewise linear is made
 so that we can represent ATFs via a sorted sequence of breakpoints, i.e., pairs $(t_i,a(t_i))$ for
 $i=1,\ldots,b$, where $b\ge 1$ and $t_1 = t_{\min}$, $t_b=t_{\max}$, and $t_i < t_{i+1}$ for $i=1,\ldots,b-1$
 (then we say that $a$ has $b$ \emph{breakpoints} and $b$ \emph{segments}).
 We call $(t_1,a(t_1)),\ldots,(t_{b-1},a(t_{b-1}))$ the \emph{inner} breakpoints of $a$.
 We assume that the slope changes at every inner breakpoint (otherwise it
is redundant).

 Arrival time functions can also be used to model time windows and the time needed for a pickup or a delivery.
 For example, if a delivery takes $0.1$ time units and must start between $4$ and $5$, the corresponding arrival time function is $a : (-\infty, 5] \to \mathbb{R}$ with $a(t) = \max\{4, t\} + 0.1$ for $t \le 5$.
 Thus we assume that we have ATFs not just for every arc in our road network but also for every action.

 \subsection{Basic operations}

 Representing time-dependent travel times by ATFs facilitates simple and efficient implementations of several key operations.
  Most of these operations are rather straight-forward and well-known; we will give concise descriptions for completeness.
 Similar algorithms are employed in previous works on time-dependent routing problems such as \cite{Visser}.
 However, we will devise a faster algorithm to compute the pointwise minimum.

 Any vehicle routing algorithm must know how long it takes a vehicle to get from a point $p$ to another point $q$.
 For any fixed path $P$ from $p$ to $q$ in $R$, we can obtain this information by composing the arrival time functions of the arcs of $P$.
 Similarly, if we perform an action $s$ (e.g., pickup or delivery) at the location $q$, we can obtain the ATF that models traveling along the path $P$ and performing $s$ by composing the ATF of $P$ with the ATF of $s$.
 Composing two ATFs is easy:

 \begin{proposition}\label{prop:composing_atfs}
  Given arrival time functions $a_1$ and $a_2$ with $b_1$ and $b_2$ breakpoints respectively,
  we can compute the composition $a_2 \circ a_1$ in $O(b_1 + b_2)$ time.
  Moreover, $a_2 \circ a_1$ is itself an arrival time function with at most $b_1 + b_2 - 1$ breakpoints.
 \end{proposition}

 \begin{proof}
  The function $a_2 \circ a_1$ can be computed using a simple sweep through $\dom(a_1)$.
  Note that in general this only gives us an $O(b_1 \cdot b_2)$ time algorithm.
  However, since $a_1$ is non-decreasing, when we sweep $t$ through $\dom(a_1)$, $a_1(t)$ sweeps through $\dom(a_2)$ only once.
  Hence, only $O(b_1 + b_2)$ time is needed to perform the sweep.
  Moreover, this also shows that $a_2 \circ a_1$ has at most $(b_1-1) + (b_2-1)$ inner breakpoints
  since each inner breakpoint of the composition must originate from a distinct inner breakpoint of $a_1$ or $a_2$.
 \end{proof}

For composing a sequence of ATFs we have the following.

 \begin{corollary}\label{cor:composing_atfs}
  Given arrival time functions $a_1, \ldots, a_k$ such that $a_i$ has $b_i$ breakpoints ($i=1,\ldots,k$),
  we can compute $a_k \circ \dots \circ a_1$ in $O((\log{k}) \sum_{i = 1}^{k}{b_i})$ time.
  Moreover, the resulting ATF has at most $1+\sum_{i = 1}^{k}{(b_i-1)}$
breakpoints.
 \end{corollary}

 \begin{proof}
  By induction on $k$.
  Using the associativity of functional composition, we can apply the induction hypothesis
  to compute $a_k \circ \dots \circ a_{\lceil\frac{k}{2}\rceil + 1}$ and $a_{\lceil\frac{k}{2}\rceil} \circ \dots \circ a_1$, and then apply Proposition~\ref{prop:composing_atfs}.
  Every breakpoint of an original ATF $a_i$ takes part in $O(\log k)$ many compositions.
 \end{proof}

 For a fixed path $P$, we denote the composition of the ATFs of its arcs by $a_P$.
 Of course, between two given points $p$ and $q$, different paths can be optimal at different times of the day. See Figure~\ref{fig:alternative_fastest_paths} on page~\pageref{fig:alternative_fastest_paths} for an example.
 Thus, the optimal arrival time between $p$ and $q$ is modeled by the function
 \[
  a_{p, q}(t) \coloneqq \min \{a_P(t) : P \text{ is a $p$-$q$-path}\}.
 \]
 The pointwise minimum of multiple ATFs can also be computed efficiently as we will see shortly.

 However, a road network can contain exponentially many paths between two
points $p$ and~$q$.
 Foschini, Hershberger and Suri \cite{FHS14} showed that even if the arc ATFs are all affine,
 the function $a_{p, q}$ may have $n^{\Omega(\log{n})}$ breakpoints,
 where $n$ is the number of vertices of the digraph~$R$.

 Thus, an exact computation of the relevant ATFs may be practically infeasible.
 We preprocess the road network using contraction hierarchies (cf.\ Section~\ref{sec:preprocessing}) and approximate ATFs where necessary.
Any such approach relies crucially on the ability to compute the pointwise minimum of arrival time functions very efficiently.

 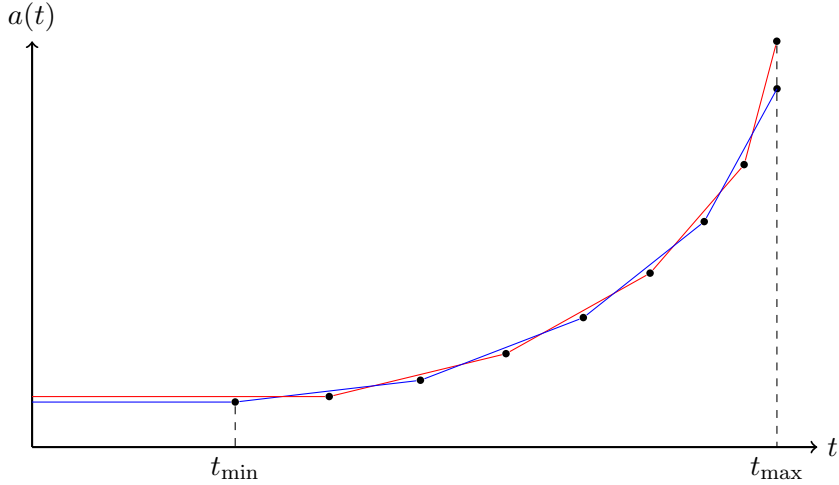
\begin{figure}[tb]
  \begin{center}
   \begin{tikzpicture}[yscale=1.2, xscale=1.8]
    \foreach \theta [count=\i] in {-80, -60, ..., 0} {
     \node[circle, fill, inner sep=1pt] (a\i) at (\theta:4) {};
    }
    \foreach \theta [count=\i] in {-90, -70, ..., -30} {
     \node[circle, fill, inner sep=1pt] (b\i) at (\theta:4) {};
    }

     \node[circle, fill, inner sep=1pt, opacity=0] (b5) at (-10:4) {};

    \foreach \i [count=\j from 2] in {1, ..., 4} {
     \draw[red, -] (a\i) -- (a\j);
    }

    \foreach \i [count=\j from 2] in {1, ..., 3} {
     \draw[blue, -] (b\i) -- (b\j);
    }

    \draw[thick, ->] (-1.5, -4.5) -- (4.3, -4.5);
    \draw[thick, ->] (-1.5, -4.5) -- (-1.5, 0);

    \draw[blue, -] ($(-1.5, -4.5)!(b1)!(-1.5, 0)$) -- (b1);
    \draw[red, -] ($(-1.5, -4.5)!(a1)!(-1.5, 0)$) -- (a1);

    \coordinate (tmin) at ($(-1.5, -4.5)!(b1)!(4.3, -4.5)$);
    \coordinate (tmax) at ($(-1.5, -4.5)!(a5)!(4.3, -4.5)$);

    \coordinate (i1) at (intersection of tmax--a5 and b4--b5);
    \node[circle, fill, inner sep=1pt] at (i1) {};
     \draw[blue, -] (b4) -- (i1);

    \draw[-, dashed] (b1) -- (tmin);
    \draw[-, dashed] (a5) -- (tmax);

    \node[below] at (tmin) {$t_{\min}$};
    \node[below] at (tmax) {$t_{\max}$};

    \node[right] at (4.3, -4.5) {$t$};
    \node[above] at (-1.5, 0) {$a(t)$};
   \end{tikzpicture}
  \end{center}
  \caption{Taking the minimum of two ATFs may increase the number of inner breakpoints by a factor of 2.
   \label{fig:atf_min_growth}
  }
 \end{figure}

 \begin{proposition}\label{prop:minimum_atfs}
  Given arrival time functions $a_1$ and $a_2$ with $b_1$ and $b_2$ breakpoints respectively,
  we can compute the pointwise minimum $\min \{a_1, a_2\}$ in $O(b_1 + b_2)$ time.
  If $\dom(a_1) = \dom(a_2)$, then $\min \{a_1, a_2\}$ is itself an arrival time function with at most $2(b_1 + b_2)-3$ breakpoints.
 \end{proposition}

 \begin{proof}
  We can use a standard sweepline algorithm to compute the points where $a_1$ and $a_2$ intersect.
  Since we assume that we are given a sorted list of the breakpoints of those functions, this can easily be implemented in linear time.
  It is also clear that $\min\{a_1, a_2\}$ is again initially constant, non-decreasing, and satisfies $\min\{a_1(t), a_2(t)\} \geq t$.
  The function is continuous if $\dom(a_1) = \dom(a_2)$.

  Inner breakpoints of $\min\{a_1, a_2\}$ can only appear at inner breakpoints of $a_1$ and $a_2$ or where $a_1$ and $a_2$ intersect.
  But in each interval in which both ATFs are affine they intersect at most once. Because there are at most $(b_1 - 1) + (b_2 - 1)$ inclusionwise maximal intervals in which $a_1$ and $a_2$ are affine and intersect, this concludes the proof.
 \end{proof}

 Unlike the composition of two functions, taking the minimum of two ATFs may in general increase the total number of breakpoints.
 See Figure~\ref{fig:atf_min_growth} for an example.

 However, it is known that the lower envelope (i.e., pointwise minimum) of $n$ line segments in the plane can have at most $O(n \alpha(n))$ breakpoints, where $\alpha(n)$ is the inverse Ackermann function \cite{HS86}.
This upper bound is sharp \cite{Wiernik86}.
 This implies that the number of breakpoints in the pointwise minimum of many arrival time functions is almost linear in the size of the input.
 A simple recursive approach as in Collorary~\ref{cor:composing_atfs} will give a running time of $O(m \alpha(m) \log n)$ where $n$ is the number of functions of which we are taking the minimum and $m$ is the total number of breakpoints \cite{agarwal2000davenport}.
We will now improve on this.

\subsection{A faster algorithm for computing the pointwise minimum}

In this section we give an algorithm for computing the pointwise minimum of $n$ piecewise linear functions with a total of $m$ breakpoints in time $O(m \log n)$, assuming that we have a sorted list of breakpoints for each input function.
 Note that this is a natural bound because all breakpoints need to be sorted by their x-coordinates for the computation of the pointwise minimum.
 In practice, the number of functions is usually small and so the pointwise minimum computation takes effectively linear time.
 For our algorithm we need the following lemma.

 \begin{lemma}\label{lem:atf_minimum_computation}
 Let $f_1, \dots, f_n \colon \mathbb{R} \to \mathbb{R}$ be affine functions, i.e.,
 $f_i(x) = s_i x + c_i$ for all $x \in \mathbb{R}$, where $s_i,c_i \in \mathbb{R}$ are constants $(i=1,\ldots,n)$.
 Then the pointwise minimum $f$ of the $f_i$ has at most $n - 1$ breakpoints.
 If the $f_i$ are sorted by their slopes $s_i$, then $f$ can be computed in $O(n)$ time.
\end{lemma}

\begin{proof}
 Without loss of generality, assume $s_1 < \dots < s_n$ (for each slope we only need to consider the function with smallest $c_i$).
 Since pointwise minima of concave functions are concave, the slope of $f$ never increases, which already shows the bound for the breakpoints.

 For the computation of $f$, we set $g_1 := f_1$ and iteratively compute $g_{i+1} := \min(g_i, f_{i+1})$ by the standard sweep line algorithm, but stop the sweep line when $f_{i+1}$ is not the minimum anymore.
 By the above observation this algorithm is correct. For each $i$ let $b_i$ be the number of breakpoints of $g_i$.
 Then, for $i = 1, \dots, n-1$ we need to compute at most $b_i - b_{i+1} + 2$ intersections of affine functions for the computation of $g_{i+1}$.
 Therefore, in total we need to compute at most $2(n-1) - b_n$ such intersections, which concludes the proof.
\end{proof}

The idea is now to use Lemma~\ref{lem:atf_minimum_computation} to compute the minimum of ATFs on each interval $I$ of the domain on which all ATFs are affine.
Naively, this would of course only give a running time of $O(m n)$.
However, we will show (again using Lemma~\ref{lem:atf_minimum_computation}) that we can use a preprocessing step to reduce the number of functions that may contribute to the minimum over $I$ to $O(\log n)$ on average.
In order to carry out this step efficiently, we will need the following technical lemma.

\begin{lemma}\label{lem:fast_sorting}
 Let $S=\{s_1,\dots,s_m\}$ be a set of $m$ real numbers and $M_1, \dots, M_k \subseteq \{1,\dots, m\}$ non-empty such that $\sum_{i=1}^k |M_i| = O(m \log m)$. Then the sets $\{s_j : j \in M_i\}$ can be sorted in total time $O(m \log m)$.
\end{lemma}
\begin{proof}
 For each $x \in \{1,\dots, m\}$ we create a list $L_x$ of all $i \in \{1, \dots, k\}$ such that $x \in M_i$. For this we only need to traverse each $M_i$ once.
 Then we find a permutation $\pi : \{1,\dots, m\} \to \{1,\dots,m\}$ such that $s_{\pi(1)} \le s_{\pi(2)} \le \dots \le s_{\pi(m)}$.
 Finally, traversing the lists $L_{\pi(1)}, \dots, L_{\pi(m)}$ in this order yields the desired orderings on all $M_i$.
\end{proof}

Now we obtain the main result of this section. It holds not only for ATFs, but for general piecewise linear functions,
and may have applications far beyond vehicle routing.

\begin{theorem}\label{thm:atf_minimum_computation}
Let $f_1, \dots, f_n \colon D \to \mathbb{R}$ be piecewise linear functions with $m$ breakpoints in total.
Then their pointwise minimum $f \colon D \to \mathbb{R}$ can be computed in $O(m \log n)$ time.
\end{theorem}
\begin{proof}
Let $x_0 < x_1 < \dots < x_m$ be all x-coordinates of breakpoints of the input functions $f_h \; (h=1,\dots, n)$ and $D=[x_0,x_m]$.
In order to simplify the notation, we assume that the x-coordinates of inner breakpoints are all distinct
 and the $f_h$ only intersect in isolated points and not in line segments.
    Without loss of generality we will assume that $n = 2^k - 1$ for some $k \in \mathbb{N}$.
 The $x_i$ can be sorted in $O(m \log n)$ time by merging the sorted lists of each $f_h$.

 Let $I_l := [x_{l-1}, x_{l}]$ be the $l$-th inclusionwise maximal interval in which all input functions are affine.
    From now on we assume that there are $2^k$ such intervals and show that in that case we can achieve a running time of $O(2^k k) = O(n \log n)$.
    The general case then follows by computing minima in chunks of size at most $2^k$.\\

 \textbf{Step 1:}
 For each $i\in \{1,\dots,k\}$ we subdivide the domain $D$ into intervals $I_{i,1}, \dots, I_{i, 2^i}$ as follows
 (cf.\ Figure~\ref{fig:compute_min_of_ATFs}).
 For $i=k$, we define $I_{k, j} := I_j$ for $1 \leq j \leq 2^k$.
 For $i<k$ we define the intervals $I_{i,j}$ recursively by setting $I_{i, j} := I_{(i + 1),(2j - 1)} \cup I_{(i + 1),(2j)}$ for $1 \leq j \leq 2^i$.
 We write
 \begin{equation*}
  \mathcal{I} := \left\{I_{i, j} : 1 \leq i \leq k, 1 \leq j \leq 2^i \right\}
 \end{equation*}
 to denote the set of all these intervals.
 Note that any two elements of $\mathcal{I}$ either have disjoint interior or one of the intervals is a subset of the other.

 For each $I  \in \mathcal{I}$ we compute a set $A_I \subseteq \{f_1, \dots, f_n\}$ such that
\begin{enumerate}
\item $|A_I| = 2^{k-i}$ for $I= I_{i,j}$, \label{item:size_A_I}
\item each function in $A_I$ is affine on $I$, and  \label{item:affine_on_I}
\item for each $l \in \{1,\dots, 2^k\}$, the sets $A_I$ with $I_l \subseteq
I$ are  a partition of $\{f_1, \dots, f_n\}$. \label{item:partition}
\end{enumerate}
To show that we can indeed compute such sets $A_I$ in $O(n \log n)$ time, we use a divide-and-conquer approach.

 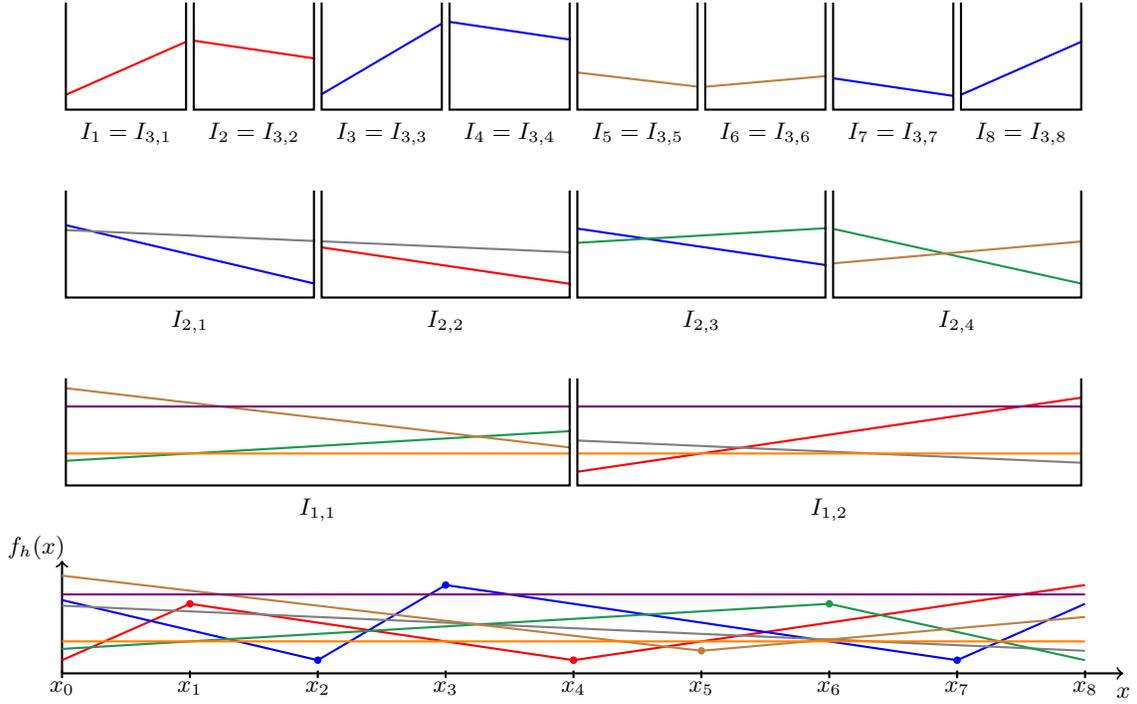
\begin{figure}[htb]
  \begin{center}
   \begin{tikzpicture}[yscale=2.5, xscale=1.7, thick]

    \draw[->] (0,0.03) to (8.3,0.03);
    \draw[->] (0,0.03) to (0,0.63);
    \node at (8.3,-0.07) {\footnotesize $x$};
    \node at (-0.2,0.7) {\footnotesize $f_h(x)$};

   \begin{scope}[red]
   \draw (0,0.1) -- (1,0.4) -- (4,0.1) -- (8,0.5);
   \node[fill,circle,inner sep=1] at (1,0.4) {};
   \node[fill,circle,inner sep=1] at (4,0.1) {};
   \draw (0,3.1) -- (1,3.4);
   \draw (1,3.4) -- (2,3.3);
   \draw (2,2.3) -- (4,2.1);
   \draw (4,1.1) -- (8,1.5);
   \end{scope}
   \begin{scope}[blue]
   \draw (0,0.42) -- (2,0.1) -- (3,0.5) -- (7,0.1) -- (8,0.4);
   \node[fill,circle,inner sep=1] at (2,0.1) {};
   \node[fill,circle,inner sep=1] at (3,0.5) {};
   \node[fill,circle,inner sep=1] at (7,0.1) {};
   \draw (0,2.42) -- (2,2.1);
   \draw (2,3.1) -- (3,3.5);
   \draw (3,3.5) -- (4,3.4);
   \draw (4,2.4) -- (6,2.2);
   \draw (6,3.2) -- (7,3.1);
   \draw (7,3.1) -- (8,3.4);
   \end{scope}
   \begin{scope}[gray]
   \draw (0,0.39) -- (8,0.15);
   \draw (0,2.39) -- (2,2.33);
   \draw (2,2.33) -- (4,2.27);
   \draw (4,1.27) -- (8,1.15);
   \end{scope}
   \begin{scope}[darkgreen]
   \draw (0,0.16) -- (6,0.4) -- (8,0.1);
   \node[fill,circle,inner sep=1] at (6,0.4) {};
   \draw (0,1.16) -- (4,1.32);
   \draw (4,2.32) -- (6,2.4);
   \draw (6,2.4) -- (8,2.1);
   \end{scope}
   \begin{scope}[brown]
   \draw (0,0.55) -- (5,0.15) -- (8,0.33);
   \node[fill,circle,inner sep=1] at (5,0.15) {};
   \draw (0,1.55) -- (4,1.23);
   \draw (4,3.23) -- (5,3.15);
   \draw (5,3.15) -- (6,3.21);
   \draw (6,2.21) -- (8,2.33);
   \end{scope}
   \begin{scope}[violet]
   \draw (0,0.45) -- (8,0.45);
   \draw (0,1.45) -- (4,1.45);
   \draw (4,1.45) -- (8,1.45);
   \end{scope}
   \begin{scope}[orange]
   \draw (0,0.2) -- (8,0.2);
   \draw (0,1.2) -- (4,1.2);
   \draw (4,1.2) -- (8,1.2);
   \end{scope}

   \fill[white] (0.98,3) rectangle (1.02,3.6);
   \fill[white] (2.98,3) rectangle (3.02,3.6);
   \fill[white] (4.98,3) rectangle (5.02,3.6);
   \fill[white] (6.98,3) rectangle (7.02,3.6);
   \fill[white] (1.98,2) rectangle (2.02,3.6);
   \fill[white] (5.98,2) rectangle (6.02,3.6);
   \fill[white] (-0.02,1) rectangle (0.02,3.6);
   \fill[white] (3.98,1) rectangle (4.02,3.6);
   \fill[white] (7.98,1) rectangle (8.02,3.6);

      \foreach \j in {1,...,8} {
      \draw (\j-0.97,3.6) -- (\j-0.97,3.03) -- (\j-0.03,3.03) -- (\j-0.03,3.6);
      \node at (\j-0.5,2.9) {\footnotesize $I_{\j}=I_{3,\j}$};
    }
   \foreach \j in {1,...,4} {
      \draw (2*\j-1.97,2.6) -- (2*\j-1.97,2.03) -- (2*\j-0.03,2.03) -- (2*\j-0.03,2.6);
      \node at (2*\j-1,1.9) {\footnotesize $I_{2,\j}$};
    }
   \foreach \j in {1,...,2} {
      \draw (4*\j-3.97,1.6) -- (4*\j-3.97,1.03) -- (4*\j-0.03,1.03) -- (4*\j-0.03,1.6);
      \node at (4*\j-2,0.9) {\footnotesize $I_{1,\j}$};
    }
    \foreach \j in {0,...,8} {
      \draw (\j, 0.005) -- (\j,0.045);
      \node at (\j,-0.05) {\footnotesize $x_{\j}$};
    }
   \end{tikzpicture}
  \end{center}
  \caption{The intervals $I_{i,j}$ (for $k=3$ and $n=7$) in the proof of Theorem~\ref{thm:atf_minimum_computation},
   and a possible assignment of the seven piecewise linear functions shown at the bottom to these intervals.
    \label{fig:compute_min_of_ATFs}
  }
 \end{figure}

    Let $T(n)$ be the number of steps needed to carry out the construction of all $A_I$.
Since the number of breakpoints in the interiors of $I_{1, 1}$ and $I_{1, 2}$ is each less than $2^{k - 1}$, we can easily find sets $A_{I_{1, 1}}$ and $A_{I_{1, 2}}$ that fulfill properties \ref{item:size_A_I} and \ref{item:affine_on_I} in $C n$ steps for some $C > 0$.
In order to satisfy property~\ref{item:partition}, we can then recurse on $I_{1, 1}$ with the set of functions $\{f |_{I_{1, 1}} \mid f \notin A_{I_{1, 1}}\}$ and likewise for $I_{1, 2}$.
    Thus, $T(n) \leq 2 T(\tfrac{n}{2}) + C n$ and so we have $T(n) = O(n \log n)$.
    Note also that the total number of elements in all $A_I$ is also $O(n \log n)$.

 \textbf{Step 2:} We use Lemma~\ref{lem:fast_sorting} to sort all the $A_I$ by the slopes of the affine segments. Then, for each $I \in \mathcal{I}$ we can compute
 the pointwise minimum $f_I$ of the input functions in $A_I$ on $I$ by Lemma~\ref{lem:atf_minimum_computation}, using $O(n \log n)$ time in total. This implies for each $l \in \{1, \dots, 2^k\}$ and each $x \in I_l$:
 \begin{equation}\label{thm:atf_minimum_computation:char}
 \begin{split}
  f(x) & = \min\{f_h(x) \ | \ 1 \leq h \leq n\} \\\
  & = \min\{ \min\{g(x) \ | \ g \in A_I\} \ | \ I_l \subseteq I \in \mathcal{I} \} \\\
  & = \min\{ f_I(x) \ | \ I_l \subseteq I \in \mathcal{I}\}. \\\
  \end{split}
 \end{equation}

 \textbf{Step 3:} For each $ l \in \{1,\dots, 2^k\}$ let us consider the functions $f_I$ with $I_l \subseteq I \in \mathcal{I}$. Let $S_l$ be the set of all affine segments of the input functions $f_h$
 that represent such a function $f_I$ with $I_l \subseteq I \in \mathcal{I}$ on a non-empty open interval in $I_l$. Since $f_I$ consists of at most $|A_I|$ many affine segments  due to Lemma~\ref{lem:atf_minimum_computation}, we get
 \begin{equation*}
     \sum_{l=1}^{2^k} |S_l| \leq \sum_{I \in \mathcal{I}} \bigl( |A_I| +  |\{l \ | \ I_l \subseteq I\}| - 1 \bigr) = O(n \log n).
 \end{equation*}
 Therefore, we can use Lemma~\ref{lem:fast_sorting} again to sort all $S_l$ by the slopes of the affine segments. Computing the pointwise minimum of the segments in $S_l$ on $I_l$
 for each $1 \leq l \leq 2^k$ yields the desired result $f$ by equation (\ref{thm:atf_minimum_computation:char}).
 The running time again follows from Lemma~\ref{lem:atf_minimum_computation}.
\end{proof}

 \subsection{Approximating arrival time functions}

The operations of composing or forming the pointwise minimum of ATFs can lead to ATFs with many breakpoints.
In order to speed up our algorithms, it is essential to reduce the number of breakpoints, which we will do by approximating ATFs.
In order to guarantee the feasibility of our solutions, the approximated ATF should be an upper bound on the original ATF. Moreover, since we have to approximate ATFs millions of times, the algorithm must be very efficient.

For these reasons we adapt the classic polygonal line simplification algorithm due to Imai and Iri \cite{II86}.
Given $\varepsilon >0$ and a continuous, piecewise linear function $f$ defined on a compact interval,
the algorithm computes in linear time another such function $g$ with $f \leq g \leq f + \varepsilon$ such that $g$ has a minimum number of breakpoints.
However, the Imai--Iri algorithm is not designed to preserve monotonicity in general.
Moreover, the domain of ATFs begins at $-\infty$.
One can see that monotonicity can be easily restored in a linear time post-opt routine (as done for example in \cite{Batz}).
However, we will show that with an appropriate choice of the last breakpoint, the algorithm will preserve monotonicity for free.

The Imai--Iri algorithm interprets the problem of approximating $f$ as the problem of traversing the polygonal corridor
between $f$ and $f + \varepsilon$ with the least number of breakpoints.
As the Imai--Iri algorithm, our algorithm computes the area in which the $i$th breakpoint can be, for $i=1,2,\dots$.
The ``right'' boundary of this area will be denoted by $W_i$.
The $i$th breakpoint $q_i$ of $g$ will always be chosen on $W_i$.
See Figure~\ref{fig:visibility} for an example.

 \begin{figure}[htb]
  \begin{center}
   \begin{tikzpicture}[xscale=2, yscale=1.2]

    \foreach \x/\y [count=\i] in {0.2/0, 1.5/0, 3.2/0.5, 4.1/2, 5/2.5, 5.3/3.4, 6.3/3.6, 6.9/3.6, 7.2/4.1} {
       \coordinate (p\i) at (\x, \y);
       \coordinate (r\i) at (\x, {\y+0.75});
    }
      \foreach \i in {2, ..., 9} {
       \node[circle, fill, inner sep=0.1em] at (p\i) {};
       \node[circle, fill, inner sep=0.1em] at (r\i) {};
    }
    \foreach \i [count=\j from 2] in {1, ..., 8} {
     \draw[-] (p\i) -- (p\j);
     \draw[-] (r\i) -- (r\j);
    }

    \draw[-, thick] (p9) -- (r9);

       \coordinate (i1) at (intersection of r1--r2 and p3--p4);
       \coordinate (i2) at (intersection of r3--p6 and r6--r7);
       \coordinate (i3) at (intersection of p4--r5 and r6--r7);
       \coordinate (i4) at (intersection of p4--r3 and p5--p6);
       \coordinate (i5) at (intersection of p6--r8 and p9--r9);
       \coordinate (q0) at (intersection of p1--r1 and r2--i1);
       \coordinate (q1) at (intersection of r1--r2 and r3--p6);
       \coordinate (q2) at (intersection of p4--p6 and p6--p7);
       \coordinate (q3) at (i5);

       \fill[blue,opacity=0.3] (p1) -- (p2) -- (p3) -- (i1) -- (r2) -- (r1) -- cycle;
       \draw[blue, very thick] (r2) -- (i1);
       \fill[darkgreen,opacity=0.3] (r2) -- (i1) -- (p4) -- (i4) -- (p6) -- (i2) -- (i3) -- (r5) -- (r4) -- (r3) -- cycle;
       \draw[darkgreen, very thick] (p6) -- (i2);
       \fill[orange,opacity=0.3] (q2) -- (p7) -- (p8) -- (p9) -- (i5) -- (r8) -- (r7) -- (i2) -- cycle;

       \draw[-, red, dashed, very thick] (q0) -- (q1) -- (q2) -- (q3);
       \node[circle, red, fill, inner sep=0.12em] at (q1) {};
       \node[below left=-0.5mm, red] at (q1) {\small $q_1$};
       \node[circle, red, fill, inner sep=0.12em] at (q2) {};
       \node[below right=-0.5mm, red] at (q2) {\small $q_2$};
       \node[circle, red, fill, inner sep=0.12em] at (q3) {};
       \node[right=0mm, red] at (q3) {\small $q_3$};

       \node[blue,right=1mm] at (i1) {\small $W_1$};
       \node[darkgreen,above] at (i2) {\small $W_2$};

       \node[below] at (2.2, 0.25) {\small $f$};
       \node[above] at (2.2, 1) {\small $f + \varepsilon$};

       \node[below] at (p2) {\scriptsize $p_1$};
       \node[above] at (r2) {\scriptsize $p_1^+$};
       \node[below right=-0.5mm] at (p9) {\scriptsize $p_b$};
       \node[above right=-0.7mm] at (r9) {\scriptsize $p_b^+$};
   \end{tikzpicture}
  \end{center}
  \caption{Illustrating Definition~\ref{def:visibility} and Algorithm~\ref{alg:imai_iri_atfs}.
      The algorithm progresses through the corridor between $f$ and $f+\varepsilon$ by constructing visibility polygons
      (shaded blue, green, orange) and their windows on the right boundary.
      The blue area is weakly visible from infinitely far to the left; $P_1$ is the remaining polygon.
      The green area is weakly visible from the blue area, and in fact from the window $W_1$.
      There are three invisibility polygons from $W_1$; the relevant one, $P_2$, is to the right of the window $W_2$
      (the window from $W_1$ to the line segment $\overline{p_bp_b^+}$ in $P_1$).
      The orange area is the part of $P_2$ that is weakly visible from $W_2$.
      A possible output $g$ is the ATF shown red and dashed, with breakpoints $q_1,q_2,q_3$.
   \label{fig:visibility}
  }
 \end{figure}
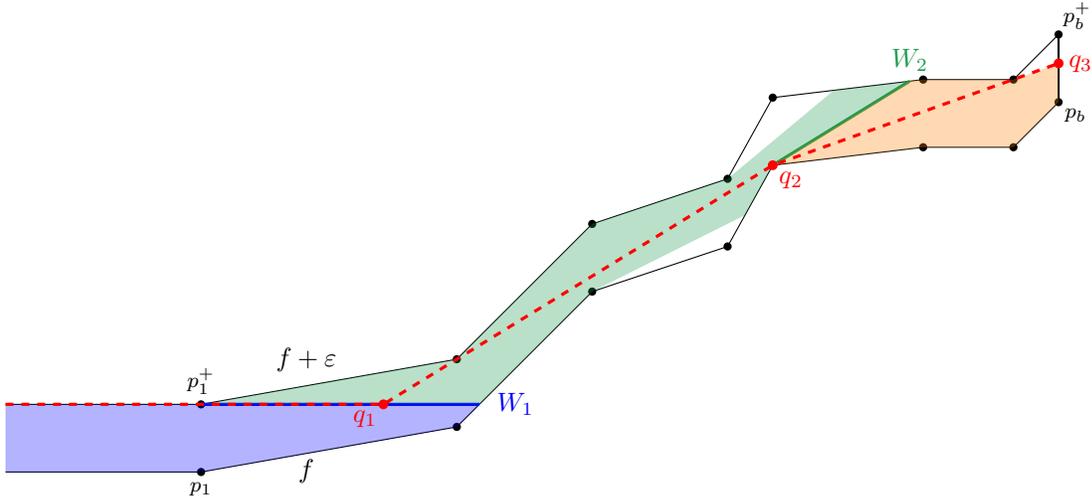

To formalize this idea we use the following notion.

 \begin{definition}\label{def:visibility}
  Let $P$ be a polygon and $A, B \subseteq P$ two polygons inside $P$ (we allow points and line segments as degenerate polygons).
  Then $B$ is \emph{weakly visible} from $A$ in $P$ if there exists a line segment $L \subseteq P$ with $L \cap A \neq \emptyset$ and $L \cap B \neq \emptyset$.
The set of all weakly visible points from $A$ in $P$ is called the \emph{visibility polygon} of $A$ in $P$.
  $P$ is partitioned into the visibility polygon and possibly several \emph{invisibility polygons}.
  If $B$ is not weakly visible from $A$ in $P$, we will call the line segment $W$ which borders the visibility polygon of $A$ and the invisibility
polygon containing $B$ the \emph{window} from $A$ to $B$ in $P$.
 \end{definition}

 \begin{algorithm}[htb]
  \caption{Approximating an ATF with fewer breakpoints.\label{alg:imai_iri_atfs}}
{\small
  \begin{tabularx}{\textwidth}{l X}
   \textbf{Input: }& an ATF $f : [-\infty, t_{\max}] \to \mathbb{R}$ and $\varepsilon > 0$ \\
   \textbf{Output: }& an ATF $g : [-\infty, t_{\max}] \to \mathbb{R}$ with $f \le g \le f + \varepsilon$ and a minimum number of breakpoints
  \end{tabularx}
  \begin{algorithmic}[1]
   \State Let $(t_{\min}, f(t_{\min}))$ be the first breakpoint of the ATF $f$
   \If{$f(t_{\max})\le f(t_{\min})+\varepsilon$}
   \State \Return the ATF $g$ given by $g(t):=f(t_{\max})$ for all $t\le t_{\max}$\label{algo:imai-iri-return-constant}
   \Else
   \State Let $p_1, \ldots, p_b$ be the breakpoints of $f$ and $p^+_1, \ldots, p^+_b$ the breakpoints of $f + \varepsilon$
   \State $t':= \min\{t: f(t)=f(t_{\min})+\varepsilon\}$ and $p':=(t',f(t'))$
   \State $W_1:= \overline{p_1^+ p'}$\label{alg:imai_iri:s1}
   \State $P_1 \coloneqq \text{polygon bounded by } p', p^+_1, \ldots, p^+_b, p_b, \ldots, p_j$, where $p_j$ is the first breakpoint after $t'$
   \State $i \coloneqq 1$
   \While{$\overline{p_b p^+_b}$ is not weakly visible from $W_i$ in $P_i$}
    \State $W_{i + 1} \coloneqq \text{window from } W_i \text{ to } \overline{p_b p_b^+} \text{ in } P_i$
    \State $P_{i + 1} \coloneqq \text{invisibility polygon from } W_i \text{ in } P_i \text{ containing } \overline{p_b p^+_b}$
    \State $q_i \coloneqq \text{intersection between the segment $W_i$ and the line containing the segment } W_{i + 1} $
    \State $i \coloneqq i + 1$
   \EndWhile
   \State Let $q_i \in W_i$ and $q_{i+1}\in \overline{p_b p^+_b}$ be visible from each other in $P_i$ s.t.\ $\overline{q_i q_{i + 1}}$ is non-decreasing\label{alg:imai_iri:pick_qi_qip1}
   \State \Return the ATF $g$ that is given by the breakpoints $q_1, \ldots, q_{i+1}$
   \EndIf
  \end{algorithmic}
}
 \end{algorithm}

A formal description of our algorithm is given by Algorithm~\ref{alg:imai_iri_atfs}.
Exactly as in \cite{II86}, the algorithm can be implemented to run in linear time using a sweep-line procedure.
To show correctness, we need two lemmata.
The first one follows the argument by Imai and Iri~\cite{II86}.

\begin{lemma}
If Algorithm~\ref{alg:imai_iri_atfs} outputs a function $g$ with $k$ breakpoints,
then any continuous, piecewise linear function $g'$ with $f\le g'\le f+\varepsilon$ has at least $k$ breakpoints.
\end{lemma}
\begin{proof}
For any such function $g'$, the first breakpoint cannot be in $P_1 \setminus W_1$.
By induction, the $i$th breakpoint cannot be in $P_i \setminus W_i$.
\end{proof}

Moreover, we need to show that the algorithm maintains monotonicity.

 \begin{lemma}
 In line~\ref{alg:imai_iri:pick_qi_qip1} of Algorithm~\ref{alg:imai_iri_atfs}, the points $q_i$ and $q_{i + 1}$ can be chosen such that the line segment $\overline{q_i q_{i + 1}}$ is non-decreasing.
The output $g$ is an ATF.
 \end{lemma}

\begin{proof}
        For the first statement, choose $q_i$ and $q_{i + 1}$ such that the
        slope of the line segment $L\coloneqq\overline{q_i q_{i + 1}}$ is maximal.
        We show that $L$ is non-decreasing.
        First note that $L$ must intersect $f + \varepsilon$ somewhere.
        Otherwise we could shift $q_{i + 1}$ upwards by some tiny amount and get a contradiction to the maximality of the slope of $L$.
        Let $t$ be maximal with $(t, f(t) + \varepsilon) \in L$.
        Then there must be some point $(t', f(t')) \in L$ with $t' < t$.
        Otherwise we could simply rotate $L$ counterclockwise around $(t, f(t) + \varepsilon)$ by some tiny angle to get a contradiction once again. 
        By monotonicity of $f$ we have $f(t') \le f(t) < f(t) + \varepsilon$. Hence, $L$ has positive slope.

        We now prove that the output $g$ is monotone and hence an ATF.
        This is clear if $g$ is returned in line 3, because then $g$ is constant.
        Otherwise, let $q_j=(t_j,g(t_j))$ for $j=1, \ldots, i + 1$ be the breakpoints of $g$ and assume, for the sake of deriving a contradiction, that there is some $j < i$ such that $g(t_j) > g(t_{j +
                1})$.

        Consider now $L \coloneqq \overline{q_j q_{j + 1}}$.
        Let $t$ be maximal with $(t, y) \in W_{j + 1}$ for some $y$.
        Since $L$ is decreasing and $L$ and $W_{j+1}$ belong to the same line, $W_{j+1}$ is also decreasing.
        Note that since $f$ and $f + \varepsilon$ are both non-decreasing and $W_{j+1}$ is decreasing, we must have $y = f(t)$. 

        By the same argument as above there must be some $t' < t$ with $(t', f(t')) \in L$ as otherwise we could shift $L$ upwards or rotate it slightly
counterclockwise to find a line segment connecting $W_j$ to a point in $P_{j+1}$ (the invisibility polygon from $W_j$ in $P_j$ containing $\overline{p_b p^+_b}$).
  But this would be a contradiction by the definition of invisibility polygons.
  Now observe that since $L$ is decreasing, we get $f(t') > f(t)$ despite
$t' < t$ which is a contradiction once more.
\end{proof}

 We conclude:

  \begin{theorem}\label{thm:simplifying_atfs}
  Given an ATF $a:[-\infty,t_{\max}]\to\mathbb{R}$ with
$b$ breakpoints and $\varepsilon>0$, we can compute an ATF
  $a':[-\infty,t_{\max}]\to\mathbb{R}$ with $a(t)\le a'(t)\le a(t)+\varepsilon$ for all $t \le t_{\max}$
  such that $a'$ has the fewest possible number of breakpoints in $O(b)$ time.
  \hfill \qed
  \end{theorem}

When applying this theorem to approximate ATFs, we choose $\varepsilon$ to be a fixed percentage (e.g., 0.5\,\%)
of the minimum travel time $\min\{a(t)-t:t\le t_{\max}\}$.

  While Algorithm~\ref{alg:imai_iri_atfs} guarantees the minimum number of breakpoints for a fixed error $\varepsilon$,
  it will often output solutions with fairly poor average error for this minimum number of breakpoints (Figure~\ref{fig:visibility} shows an example).
  Therefore, we employ a simple post-processing heuristic that moves individual breakpoints or line segments optimally,
  keeping the remaining ATF fixed and maintaining the properties guaranteed by Theorem~\ref{thm:simplifying_atfs}.

 \section{Evaluating and updating tours}
 \label{sec:upd-tours}

   A task that our overall algorithm performs very often is to insert an item (with its pickup and its delivery) into or remove an item from an existing tour.
  Even more often we want to evaluate the effect of such a change, i.e., whether it is possible at all to insert the item and if so at which cost, or
  how much cheaper a tour becomes when we remove an item.
  Only a small subset of the evaluated changes will actually be performed.
  Therefore we maintain data structures that allow fast queries and update them whenever we insert or remove an item.
  These data structures also enable us to efficiently evaluate exchanges of sequences of consecutive items between existing tours.

  Whenever we evaluate an insertion or removal of an item (or a sequence of items),
  we first compute the ATF of the resulting tour and then perform a scheduling in order to determine the cost of that tour.
  In Section~\ref{subsec:cheapest_insert}, we describe our data structure for maintaining ATFs of tours and sections of tours.
  In Section~\ref{subsec:scheduling}, we describe our scheduling algorithm.

\subsection{Maintaining arrival time functions of tours}\label{subsec:cheapest_insert}

  Recall that a tour sequence is a sequence of actions.
  For each action $i$ we have an arrival time function $a_i$.
  If we arrive at the location of action $i$ at time $t$, the time $a_i(t)$ is the earliest possible
  arrival time at the location of the following action after having performed action $i$ and traveled to the next location.
  This can include a waiting time if action $i$ cannot be performed immediately.

  Our data structure is built on ATFs $a_1, \dots, a_n$ and allows us to efficiently compute the compositions $a_{i, j} := a_{j} \circ \dots \circ a_{i+1}$ for all $0 \leq i < j \leq n$.
  While for some simple vehicle routing problems, the ATFs $a_{0,i}$ and $a_{i,n}$ for $i=1,\dots,n$ are sufficient to evaluate insertions and removals of single items, our data structure works for general pickup and delivery problems and also allows for efficiently evaluating larger changes such as swapping longer segments between tours.

  Note that if $a_i$ has $b_i$ breakpoints for each $i$, the functions $a_{i, j}$ have at most $b_{i, j} := \sum_{h=i+1}^{j} b_h$  breakpoints (cf.\ Proposition~\ref{prop:composing_atfs}).
 Therefore, we want to bound the query time of some $a_{i,j}$ roughly by $O(b_{i,j})$, but also initialization and update time for the data structure is relevant.

  A first efficient data structure was given by Visser and Spliet \cite{Visser},
  storing the actions in a binary search tree (see Figure~\ref{fig:tree_cheapest_insert}).
  Let $b := \sum_{i=1}^{n} b_i$.

   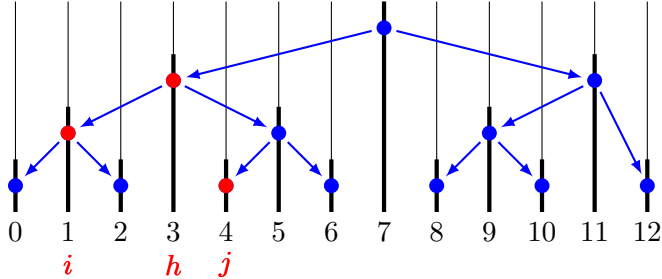
\begin{figure}[htb]
  \begin{center}
  \begin{tikzpicture}[scale=0.7]
   \foreach \x in {0, ..., 12} {
      \draw (\x,0) node[below] {\x} -- (\x,4);
   }
   \foreach \x in {7} {
      \draw[ultra thick] (\x,0) -- (\x,4);
      \node[circle, fill, blue, inner sep = 2, outer sep = 2] (p\x) at (\x,3.5) {};
   }
   \foreach \x in {3,11} {
      \draw[ultra thick] (\x,0) -- (\x,3);
      \node[circle, fill, blue, inner sep = 2, outer sep = 2] (p\x) at (\x,2.5) {};
   }
   \foreach \x in {1,5,9} {
      \draw[ultra thick] (\x,0) -- (\x,2);
      \node[circle, fill, blue, inner sep = 2, outer sep = 2] (p\x) at (\x,1.5) {};
   }
   \foreach \x in {0,2,4,6,8,10,12} {
      \draw[ultra thick] (\x,0) -- (\x,1);
      \node[circle, fill, blue, inner sep = 2, outer sep = 2] (p\x) at (\x,0.5) {};
   }
   \foreach \x/\y in {7/3,7/11, 3/1,3/5,11/9,11/12, 1/0, 1/2, 5/4, 5/6, 9/8, 9/10} {
      \draw[->, >=latex, thick, blue] (p\x) -- (p\y);
   }
   \foreach \x in {1,4,3}{
     \node[circle, fill, red, inner sep = 2, outer sep = 2] at (p\x) {};
   \node[red, thick] () at (1,-1) {$i$};
   \node[red, thick] () at (4,-1) {$j$};
   \node[red, thick] () at (3,-1) {$h$};
      }
  \end{tikzpicture}
  \end{center}
  \caption{A balanced binary search tree.
  To compute $a_{i,j}$ we find the lowest common ancestor $h$ of $i$ and $j$ and compose $a_{i,h}$ and $a_{h,j}$.
  These ATFs are precomputed because $h$ is an ancestor of $i$ and $j$.
  \label{fig:tree_cheapest_insert}
  }
  \end{figure}

  \begin{theorem}[Visser, Spliet \cite{Visser}]\label{thm:bst_data_structure}
  With $O(n \log n)$ compose operations and $O(nb)$ total initialization time, we can construct a data structure that allows the computation of any $a_{i, j}$
   for $0 \leq i < j \leq n$ in $O(b_{i,j} + \log n)$ time with at most one compose operation.
  \end{theorem}
  \begin{proof}
   We build a balanced binary search tree on $\{0, \dots, n\}$ and precompute all $a_{i,j}$ such that $i$ is an ancestor or a descendant of $j$ in that tree.
We say that a vertex is at level $k$ if it has distance $k$ from the root in that tree.
The tree has levels $k=0,\dots,\lfloor \log_2(n+1) \rfloor$.

For each level $k$ we consider all vertices at level $k$, and for each such vertex $h$ and each descendant $i$ of $h$ we compute $a_{h,i}$ (if $h<i$) or $a_{i,h}$ (if $i<h$).
For a fixed vertex $h$, let $p$ and $q$ be the smallest and largest descendant of $h$, respectively.
Then we compute the functions $a_{i,h}$ for $p\le i <h$ in the order of decreasing $i$ and we compute the functions $a_{h,i}$ for $h < i \le q$ in the order of increasing $i$.
This requires a single composition for the computation of each of these functions and hence we need in total $q-p$ compositions for each fixed vertex $h$.
Because the intervals $[p,q]$ are disjoint for the different vertices $h$ on the same level $k$, we need a total of at most $n$ compose operations for each level.

For computing a single ATF $a_{h,i}$ (for $h<i$; the case $i<h$ is symmetric), the computing time is $b_{h,i}$
because $a_{h,i}$ is computed by composing $a_{h,i-1}$ (with at most $b_{h,i-1}$ breakpoints) and $a_i$ (with at most $b_{i}$ breakpoints),
leading to a running time of $O(b_{h,i-1} + b_{i}) = O(b_{h,i})$, which is $O(b_{p,q})$.
For each vertex $h$,  we compute at most $\frac{n}{2^{k-1}}$ such functions.
Therefore, because the intervals $[p,q]$ are disjoint for the different vertices $h$ at level $k$, the total running time at level $k$ is $O(\frac{nb}{2^k})$.

Given some $0 \leq i < j \leq n$, we can find the lowest common ancestor $h$ of $i$ and $j$ in $O(\log n)$ time and compute $a_{i,j}$ by composing $a_{i, h}$ and $a_{h, j}$.
  \end{proof}

When one ATF $a_i$ changes, we can update this data structure with $O(n)$ compose operations: we only need to update the functions $a_{h,j}$ where $h < i \le j$ and $h$ is an ancestor or descendant of $j$.
Then the search tree property implies that $h$ or $j$ is an ancestor of $i$.
The ancestor of $i$ at level $k$ has at most $\frac{n}{2^{k-1}}$ descendants.
This implies that we need to update at most $4n$ functions $a_{h,j}$.

When an action $i$ is removed from the tour, we can replace $a_i$ by the identity function and apply the above update.
When a new action is inserted, we can insert a node into the search tree without rebalancing, and only after $O(\log n)$ insert or removal steps we recompute the data structure from scratch.
This does not change any of the above-mentioned running time bounds, but saves time in practice.

In practice it is also useful to store the ATFs $a_{0,i}$ and $a_{i,n}$  for all $i=1,\dots,n-1$ because these initial and final sequences are queried often.
Again, this does not increase any of the above running time bounds.

  If we allow more than one compose operation for a query, the initialization time (and the update time) can be improved further by building the data structure on small segments separately.

  \begin{theorem}\label{thm:bst_data_structure_iterative}
   Let $k \in \mathbb{N}$.
   With $O(n \log n)$ compose operations and $O(k n^{\frac{1}{k}} b)$ total initialization time, we can construct
   a data structure that allows the computation of any $a_{i, j}$ for $0 \leq i < j \leq n$ in $O(b_{i, j} \log (k+1) + \log n)$ time with at most $2k-1$ compose operations, and at most $k-1$ compose operations if $i=0$ or $j=n$.
  \end{theorem}
  \begin{proof}
   We prove the statement by induction on $k$.
   The case $k=1$ follows from Theorem~\ref{thm:bst_data_structure} and by precomputing the ATFs $a_{0,i}$ and $a_{i,n}$ as described above.

   For larger $k$, we choose $p := \lceil n^{\frac{1}{k}} \rceil$ and $m:= \lceil\frac{n}{p}\rceil$.
   First, we build the data structure from Theorem~\ref{thm:bst_data_structure} on each of the sets $A_i := \{ a_{(i-1)p+1}, \dots, a_{ip} \}$ for $i=1,\ldots,m-1$ and $A_m := \{ a_{(m-1)p+1}, \dots, a_n\}$.
   We also store the ATFs of the initial and final sequences within each $A_i$.

   Second, we apply the induction hypothesis to construct a top-level data structure on ATFs $a_i'$, which are the compositions of the ATFs in $A_i$.
   More precisely, $a'_i :=  a_{(i-1)p, ip}$ for $i=1,\dots,m-1$ and $a'_m :=  a_{(m-1)p, n}$.

   The data structures for the $A_i$ can be initialized with $O(n \log p)=O(\frac{1}{k}n\log n)$ compose operations and total time $O(p b)= O(n^{\frac{1}{k}}b)$.
By the induction hypothesis, the data structure for $a_1',\dots, a_m'$ can be initialized with $O(m\log m)=O(\frac{k-1}{k}n\log n)$ compose operations
in total time $O((k-1)m^{\frac{1}{k-1}}b)= O((k-1)n^{\frac{1}{k}}b)$.
Since the constant hidden in the $O$-notation is always the same, this implies the claimed bounds for the initialization.

   For the query of some $a_{i,j}$ where $a_{i+1}$ and $a_j$ are not in the same set $A_h$, we set $s:=\lceil \frac{i}{p}\rceil$ and $t:=\lfloor \frac{j}{p}\rfloor$.
   Note that $a_{i,j}$ is the composition of the ATFs $a_{i,sp}$ and $a_{sp,tp}= a'_{s,t}$ and $a_{tp,j}$.
   By the induction hypothesis we can write $a'_{s,t}$ as the composition of $2(k-1)$ precomputed ATFs.
   The ATFs $a_{i,sp}$ and $a_{tp,j}$ are precomputed in the data structures on the sets $A_s$ and $A_{t+1}$.
   Hence we can write $a_{i,j}$ as the composition of $2k$ precomputed ATFs.
   We can find these in $O(k+\log p)=O(\log n)$ time.
   Since they have a total of $b_{i,j}$ breakpoints, the running time follows from Corollary~\ref{cor:composing_atfs}.

   For querying $a_{0,j}$, the induction hypothesis yields that $a_{0,tp}$ is the composition of $k-1$ precomputed ATFs, and $a_{tp,j}$ is precomputed in the data structure on the set $A_{t+1}$.
   So we need to compose $k$ ATFs, and the same holds for querying $a_{i,n}$ by symmetry.
  \end{proof}

This multi-level data structure also allows for lazy update, which can save time in practice.
Here is an example how our data structure allows fast evaluations:

\begin{corollary}
Let $k \in \mathbb{N}$.
 Given the data structure of Theorem~\ref{thm:bst_data_structure_iterative}, it takes at most $4k+3$ compose operations
and $O(b\log k)$ total time
 to compute the total arrival time function of a tour after inserting the pickup action and the delivery action of an item
 at given positions, where $b$ is the total number of breakpoints of the ATFs of all actions in the resulting tour.
\end{corollary}

\begin{proof}
Let $a_1,\ldots,a_n$ be the ATFs of the $n$ actions of the current tour sequence.
If we insert the pickup action after action $i$ and the delivery action after action $j$,
we need to compose $a_{j,n} \circ a_d \circ a_j' \circ a_{i,j-1} \circ a_p \circ a_i' \circ a_{0,i-1}$,
where $a_i'$ is the arrival time function for performing action $i$ and then traveling to the pickup position,
$a_j'$ is the arrival time function for performing action $j$ and then traveling to the delivery position,
$a_p$ is the arrival time function for performing the pickup action and then traveling to position of action $i+1$,
and $a_d$ is the arrival time function for performing the delivery action
and then traveling to the position of action $j+1$.
All of these four ATFs are precomputed.
By Theorem~\ref{thm:bst_data_structure_iterative}, computing $a_{0,i-1}$ and $a_{j,n}$ requires $k-1$ compose operations each, and computing $a_{i,j-1}$ requires $2k-1$ compose operations.
Hence the total number of compose operations is $4k+3$.
The running time bound follows from Corollary~\ref{cor:composing_atfs}.
\end{proof}

  \subsection{Tour scheduling}\label{subsec:scheduling}

  In the previous section we have discussed how we can efficiently maintain the total arrival time function $a \colon (-\infty, t_{\max}] \rightarrow \mathbb{R}$ for every tour. In order to evaluate the quality of a tour, we will frequently need to find the optimal starting time $t_0$ at which the tour should start.
The problem of computing $t_0$ has been called the \emph{Optimal Starting Time Problem} in \cite{HASHIMOTO2008434}.

The total cost of a tour depends on various properties of the tour, many of which depend on the start time.
Examples are the tour duration and the total driving distance.
The driving distance depends on the start time of the tour because different paths are optimal at different times; see Figure~\ref{fig:alternative_fastest_paths} on page~\pageref{fig:alternative_fastest_paths}.
If the cost is proportional to the driving distance, the cost is a piecewise constant lower semi-continuous function of the start time of the tour.
(At breakpoints where two paths result in the same arrival time, we take the cheaper path; this guarantees the lower semi-continuity.)

In general we associate a piecewise constant, lower semi-continuous cost function $c_f$ with every ATF $f$ that we store.
Whenever we compose two ATFs $f$ and $g$, we compute the cost function associated with the composition via $c_{g \circ f} = c_g \circ f + c_f$.
The number of discontinuities of $c_{g \circ f}$ is at most the sum of the numbers of discontinuities of $c_f$ and $c_g$.

These piecewise constant cost functions can also be used to model further cost components, e.g., tolls or robustness costs for performing a pickup or delivery close to the end of its time window.
For instance, we might impose an artificial penalty for arriving within the last 10 minutes of a time window.
Such soft time windows will make the schedule more robust as we will see in Section~\ref{sec:results-new-benchmarks}.

The most important component of the cost function is the tour duration $a(t_0) - t_0$, where $a$ is the ATF of the tour.
To model overtime costs, we allow for a piecewise linear, continuous, and non-decreasing function $c_{\mathrm{ot}} : [0, \infty) \rightarrow \mathbb{R}$ and describe the resulting cost by $c_{\mathrm{ot}}(a(t_0)-t_0)$.
If the cost is proportional to the tour duration, $c_{\mathrm{ot}}$ is a linear function.
In practice, $c_{\mathrm{ot}}$ has very few breakpoints.

Finally, the work time cost may depend on the time of the day.
For example drivers might be paid more if they are scheduled to work outside of normal working hours.
We describe the cost per hour by a piecewise constant function $c_{\mathrm{wt}} : [t_{\min}, t_{\max}] \rightarrow \mathbb{R}$ and pay
  \[
\int_{t_0}^{a(t_0)}{c_{\mathrm{wt}}(t) \, \mathrm{d} t}
  \]
if we start at time $t_0$.

Combining the different cost components, our goal is to compute a $t_0 \in [t_{\min}, t_{\max}]$ minimizing the total cost
  \[
      c_{\mathrm{tot}}(t_0) \coloneqq c_a(t_0) + c_{\mathrm{ot}}(a(t_0) - t_0) +
\int_{t_0}^{a(t_0)}{c_{\mathrm{wt}}(t) \, \mathrm{d} t}.
  \]
Note that most of the cost components in this function can also be modeled in the cost model in \cite{HASHIMOTO2008434}.
  The algorithm proposed in \cite{HASHIMOTO2008434} also optimizes waiting time if this reduces the cost, while in our setting we allow waiting only when we arrive at a stop before the time window begins.
This seems a reasonable assumption in practice and allows us to solve the Optimal Starting Time Problem much faster:

  \begin{theorem}
      Let $a$, $c_a$, $c_{\mathrm{ot}}$ and $c_{\mathrm{wt}}$ have $b_a$, $b_{c_a}$, $b_{c_{\mathrm{ot}}}$ and $b_{c_{\mathrm{wt}}}$ breakpoints respectively.
      Then we can compute the least $t_0 \in [t_{\min}, t_{\max}]$ minimizing $c_{\mathrm{tot}}$ in $O(b_{c_{\mathrm{ot}}} b_a + b_{c_a} + b_{c_{\mathrm{wt}}})$ time.
  \end{theorem}

  \begin{proof}
      Note that $c_{\mathrm{tot}} : [t_{\min}, t_{\max}] \rightarrow
\mathbb{R}$ is itself a piecewise linear, lower semi-continuous function.
      So in order to find its minimum, it suffices to construct all of its breakpoints.
      To do this we perform a linear scan in the interval $[t_{\min}, t_{\max}]$.

      Starting with $t = t_{\min}$, we may compute $c_{\mathrm{tot}}(t) = c_{\mathrm{tot}}(t_{\min})$ easily in $O(b_a + b_{c_a} + b_{c_{\mathrm{ot}}} + b_{c_{\mathrm{wt}}})$ time.
      Now as $t$ increases, $c_{\mathrm{tot}}$ can have a breakpoint at
$t$ only if one of the following events occurs:
      \begin{enumerate}
          \item $c_a$ has a discontinuity at $t$,
          \item $a$ has a breakpoint at $t$,
          \item $c_{\mathrm{ot}}$ has a breakpoint at $a(t) - t$,
          \item $c_{\mathrm{wt}}$ has a discontinuity at $t$, or
          \item $c_{\mathrm{wt}}$ has a discontinuity at $a(t)$.
      \end{enumerate}

      Given $t$ and the relevant segments of $a$, $c_a$, $c_{\mathrm{ot}}$ and $c_{\mathrm{wt}}$, we can find the next of these events in constant time.
      Note that event~1 can occur $b_{c_a}$ times, event~2 can occur $b_a$ times, event~3 can occur $b_{c_{\mathrm{ot}}} b_a$ times, and events~4 and~5 can each occur $b_{c_{\mathrm{wt}}}$ times since $a$ is non-decreasing.
      So we can iterate through all breakpoints of $c_{\mathrm{tot}}$ in $O(b_{c_{\mathrm{ot}}} b_a + b_{c_a} + b_{c_{\mathrm{wt}}})$ time to find the
first minimizing breakpoint.
  \end{proof}

  In practice, both $c_{\mathrm{ot}}$ and $c_{\mathrm{wt}}$ tend to have very few breakpoints.
  So the running time tends to be effectively $O(b_a + b_{c_a})$.
  By approximating both $a$ and $c_a$ during the algorithm, it
is thus possible to query the optimal starting time of a tour and hence its total cost very efficiently.

\section{Constructing and improving solutions}\label{sec:our_vrp_algorithm}

In this section we explain our overall algorithm for vehicle routing with time-dependent travel times.
Of course it will call the operations described in the previous sections many times.

\subsection{Preprocessing the road network and computing fastest paths}\label{sec:preprocessing}

Before running our main algorithm, we have a preprocessing step that prepares the road network information.
This step is independent of the addresses and thus needs to be run again only when the road network or the vehicle fleet changes.
Input to preprocessing is a road network with all relevant information on restrictions and travel times.
This is modeled as a graph with undirected and directed edges.
For each edge and each vehicle type the input contains an ATF, which is infinite if the vehicle is not allowed on the road corresponding to that edge.
We end up with a separate graph for each vehicle type.
Turn restrictions (forbidden sequences of vertices) are modeled by making copies of certain vertices, depending on the predecessor(s) on a path.
With each ATF we also store a lower and an upper bound on the travel time in order to speed up certain operations later.
We also approximate ATFs at several stages.

Next, we compute a \emph{contraction hierarchy}, similarly to \cite{Batz}.
In this step, a total order of all vertices of a graph is computed.
If $v_1,\ldots,v_n$ is this order, we process $k=1,\ldots,n$ and add a shortcut edge from $v_i$ to $v_j$
(with composed ATFs) whenever $i>k$ and $j>k$ and there are edges (including previously added shortcut edges)
from $v_i$ to $v_k$ and from $v_k$ to $v_j$
such that these (might) form a fastest path from $v_i$ to $v_j$ at some time of the day (here we use lower and upper bounds).
This guarantees that there is always a fastest path whose vertex indices are first increasing and then decreasing (an \emph{up-down path}).
Occasionally, a slower path could be cheaper than the fastest path, but this effect is negligible in practice because drivers typically choose a fastest path anyway.
The contraction hierarchy typically roughly doubles the number of edges but allows for faster fastest-path queries later.
The total preprocessing time is less than 10 minutes on typical instances up to 8 million road segments and up to 2\,000 start and end addresses
(including those used to generate our benchmarks described in Section~\ref{sec:new-benchmarks})
on an AMD EPYC 7742 server. This  divides into up to 6 minutes for computing the contraction hierarchy and
up to 4 minutes for computing the matrix of arrival time functions for each pair of addresses.

When processing an actual instance with pickup and delivery addresses given by geographic coordinates, 
we first map each address to the nearest well-accessible road segment. We add a copy of the corresponding edge, subdivide it,
and put the new vertex at the beginning of our order. Keeping the original edge ensures that we still have 
a fastest up-down path for every pair of vertices.
For each such vertex $v$ we then compute a forward cone and a backward cone (as suggested by Geisberger and Sanders \cite{Geisberger2010engineering}).
The forward cone contains all vertices which are reachable from $v$ via a path with increasing indices ($i< j$ for each edge $(v_i, v_j)$ on the path) that is fastest at some time.
The backward cone contains all vertices from which $v$ is reachable via a path with decreasing indices ($i> j$ for each edge $(v_i, v_j)$ on the path) that is fastest at some time.
Moreover we store the minimum of the composed ATFs of such paths at the elements of the cones
and approximate the resulting ATF as described in Section~\ref{sec:atf}. Most cones contain a few hundred vertices and
can be computed in a few milliseconds.

When asking for a fastest path from $v$ to $w$ (a frequently needed operation), we intersect the forward cone of $v$
with the backward cone of $w$ and sort the intersection according to the sum of the two minimum travel times (a lower bound).
To compute the total ATF from $v$ to $w$, we scan the intersection in that order, compose the two ATFs and take
the minimum with the previous total ATF. We can stop when the lower bound exceeds the current maximum travel time
(which happens after fewer than ten steps normally). The total running time for this step is less than a tenth of a millisecond on average.
For each segment of the ATF we also store the time-independent part of the cost of the corresponding path (e.g., cost proportional to the driving distance).
The actual path is not stored, instead it is simply recomputed in the very end when all tour sequences are fixed;
then we also recompute optimal parking positions.

\subsection{Constructing tours}\label{subsec:clustering}

Part of an instance is a list of vehicles and a list of items to be shipped.
Each vehicle has a start address, an end address, capacities in several dimensions (e.g., weight and size), and characteristics
that define the road restrictions and speed profiles. It is also possible that start and/or end addresses can be chosen freely.
Each vehicle (type) can have a setup cost, a cost per driving distance, and a cost per used time.
Moreover, a vehicle (with a driver) is available within a certain time period, for a certain duration, possibly with an overtime option at an extra cost.
We can also have a maximum distance that an (electric) vehicle can go.
Finally, a vehicle can have sorting rules, for example that after leaving the depot all items in the vehicle must be delivered before any new pickup is allowed.

Each item has a pickup address, a pickup time window, a delivery address, and a delivery time window. Moreover, the set of
vehicles that can ship an item can be restricted. There is a time required for pickup and for delivery, and an extra time for parking:
thus we can save some time by parking only once for several pickups or deliveries nearby. These times can also depend on the vehicle type.
Finally, an item comes with a penalty that we have to pay in case we decide not to ship it.
This feature is mainly used in the case when the number of vehicles is not sufficient to handle all items; then some items can be prioritized.

In a preprocessing step we identify pairs of items that seem to fit well in the same tour (\emph{friends}),
i.e., shipping both does not cost much more than shipping one. In case we are almost sure that they
should be handled by the same vehicle (e.g., same pickup and same delivery address and compatible time windows),
we treat them as a single item, with a single parking position.
In other cases we store the friendship information and use it at several places as described below.
Moreover we greedily select a subset of items that are important (i.e., we should handle them, but it is not easy to handle them)
and not at all friends of each other.
These will serve as seeds: we open a separate tour for each seed item as a starting point of our tour construction.

Then we insert the remaining items one by one. We always guarantee to maintain a feasible solution, except that not all items belong to a tour (yet).
In each step we pick an item whose \emph{average regret} is maximum,
i.e., the difference of the expected cost of inserting it into a random tour and the cost of inserting it into the best tour.
This is similar to the approach proposed in \cite{FoisyPotvin}.
For measuring the insertion cost, we use the algorithm described in Subsection~\ref{subsec:cheapest_insert},
followed by the scheduling as described in Subsection~\ref{subsec:scheduling}.
If inserting an item into a tour is infeasible (e.g., due to violated capacity constraints or time windows),
the corresponding cost is replaced by the cost of shipping the item by a new separate tour (or a large constant if no further vehicle is available).
When we detect that more tours are needed than we have, we compute new seed tours.

Moreover, at certain stages, in particular after creating new seed tours, and more extensively at the very end,
we perform local search operations (see the next subsection) to optimize the existing tours.
After having computed different candidate solutions with parallel threads, making different random choices at various places,
we continue post-optimizing the best of these solutions in the end with all threads.

\subsection{Post-optimization by local search operations}

We implemented various local search operations, most of which are well-known from the vehicle routing literature.
The basic operations used by all of them are those described in Section~\ref{sec:upd-tours},
but we apply them often to contiguous sequences of items in a tour rather than single items.
Formally, by a \emph{sequence} we mean a set of items in a tour for which all their pickup events
are contiguous in this tour and all their delivery events are contiguous in this tour.
It is straightforward to extend the insertion and removal algorithms of Subsection~\ref{subsec:cheapest_insert} to sequences.

To evaluate the cost resulting from an insertion or removal, we need to run the scheduling algorithm from Subsection~\ref{subsec:scheduling}.
For inserting a segment into a tour we usually try several positions for the pickup sequence as well as for the delivery sequence.
If an insertion at a certain position is infeasible due to capacity constraints, incompatibility, or sorting rules, this is detected quickly.
The scheduling algorithm can determine  infeasibility due to time windows or exceeding the work time limit. 
Often our lower bounds can be used to detect infeasibility much faster even without composing ATFs and scheduling.

At any stage we maintain a feasible solution.
When we try inserting or removing a sequence of items, we do not update the main data structures until we know that we actually want to perform this change; see Section~\ref{subsec:cheapest_insert}.
Only a small percentage of the evaluated changes
will actually prove beneficial and not be discarded.

After significant changes to a tour, the tour sequence is re-optimized.
We first run the Lin-Kernighan heuristic \cite{LinKernighan} adapted to the Asymmetric Path TSP,
but only on segments in which reorderings that reduce the total travel time
are likely to maintain feasibility (with respect to sorting, capacity constraints, and the current schedule).
Next we run a (slower but accurate) dynamic programming algorithm similar to that proposed in \cite{lera2020dynamic},
but we randomly preselect a subset of (typically long) edges of the tour and force the dynamic program to keep the sequences
in between those edges contiguous.
Last, we run Helsgaun's LKH3 heuristic \cite{LKH3} adapted to our setting.

For optimizing a set of tours, we have implemented the following local search operations.
In all local search heuristics we use the friendship information to restrict the search space.

\emph{Segment swapping} operates on two tours, removes a sequence from each,
and inserts these sequences (possibly reversed) into the other tour. 
\emph{Matching} removes short (possibly empty) sequences from all tours and possibly takes a few items that are currently not served by any tour
(which can arise within the random walk framewalk to be described below),
and finds an optimal bipartite matching, assigning these sequences to the so-reduced tours (and possibly to ``outside any tour'').
\emph{Minimum mean cycle} considers a complete digraph whose vertices are the tours (or pieces of multi-trip tours in between two visits of a depot),
and in which the weight of the edge from tour A to tour B is the (possibly negative) cost of
removing a selected sequence from tour A and inserting it into tour B, ignoring feasibility constraints.
We find a cycle with minimum mean weight in this digraph. If its total weight is negative, we check whether the corresponding set of operations
is actually feasible and beneficial. If not, we increase the weight of an edge whose weight was too optimistic and iterate.
\emph{Chain opt} tries to insert items that currently do not belong to any tour, similarly to 
 \emph{ejection-chains} introduced by \cite{glover_chains} for the TSP and extended to vehicle routing problems in 
 \cite{rego_chains, sontrop_chains, soto_chains, AccorsiVigo:2020}.

\emph{Random walks} were introduced in \cite{applegate2003chained} to iterate local search for the TSP. We follow the approach described in \cite{AccorsiVigo:2020} for the vehicle routing problem: removing a random sequence from a tour,
re-inserting the items into other tours as described in Section~\ref{subsec:clustering} or in other orders as suggested in \cite{AccorsiVigo:2020},
and post-optimizing the involved tours using the above-described operations.
In a parallel branch, we also try to dissolve an entire tour and re-insert its items into other tours by the same technique and use 
ideas proposed in \cite{curtois_pdptw, sartori_pdptw} to estimate which items are probably difficult to insert. 
The two branches exchange their best found solutions regularly.

\section{Experimental results}\label{sec:experimental}

We implemented the described algorithms in C++ and integrated them into a
software package called BonnTour. To demonstrate the efficiency of BonnTour
we carried out extensive experiments.  First, we used existing
capacitated vehicle routing benchmarks without and with delivery time
windows (Section~\ref{sec:results-prevalent-cvrp}). Second, as large and realistic 
public benchmarks with time-dependent travel times were missing, we
created a new set of instances, which we make publicly available at
\url{https://gitlab.com/muelleratorunibonnde/vrptdt-benchmark}.

All our experiments were carried out on a 2-way AMD EPYC 7742 Linux server with 128 physical cores in total.
We ran multiple jobs simultaneously (each using up to 8 threads), but we never used more than 128 threads in total.

\subsection{Results on standard vehicle routing benchmarks}
\label{sec:results-prevalent-cvrp}

We tested BonnTour on the benchmarks for the capacitated vehicle
routing problem (\textsc{CVRP}) by Uchoa et al.~\cite{uchoa}, on the benchmarks
for the capacitated vehicle routing problem with time windows
(\textsc{CVRPTW}) by Gehring and Homberger~\cite{homberger},
and on the pickup-and-delivery problem with time windows (\textsc{PDPTW}) by Li and Lim~\cite{li2003metaheuristic}.
These are popular benchmark sets for state-of-the-art vehicle routing algorithms.
While there are 100 Uchoa et al.\ instances with sizes ascending from 100 to 1000 customers,
there are 300 Gehring and Homberger instances, with 60 instances in each of the groups with instance size 200, 400, 600, 800, and 1000.
All these instances have a single depot for pickup.
There are 354 Li and Lim instances ranging from 100 to 1000 tasks (each task is either a pickup or a delivery).\footnote{In their paper, Li and Lim only describe how they generated the instances with 100 tasks. All instances can be found on \url{https://www.sintef.no/projectweb/top/pdptw/}.} Each of the instances by Li and Lim is either based on one of the Gehring and Homberger instances, or on one of the 56 Solomon instances for capacitated vehicle routing with time windows with 100 customers each \cite{solomon}.
For the Uchoa et al.\ instances the objective is to minimize the total distance, whereas for the Gehring and Homberger instances and for the Li and Lim instances the main objective is to minimize the number of tours (used vehicles),
using the total distance only as second criterion.

These instances  have fixed  travel times between any pair of addresses.
Because BonnTour has the implementation overhead of maintaining time-dependent travel times and
handling many additional constraints, we cannot expect to outperform the best implementations for these special
instance types.
Here, we show that we are close to the best known solutions,  which
are reported  at \url{https://www.sintef.no/projectweb/top/vrptw/}\footnote{We checked all web addresses mentioned in this section on March 13, 2024.}
for the Gehring and Homberger instances, at \url{http://vrp.atd-lab.inf.puc-rio.br/index.php/en/}
for the Uchoa et al.\ instances, and at \url{https://www.sintef.no/projectweb/top/pdptw/} for the Li and Lim instances.
\begin{table}[htb]
        \begin{center}
        {\small
                \resizebox{\linewidth}{!}{
                        \begin{tabular}{|l r | r r | r@{} r r@{} r r | r@{} r r@{} r r|}
	\hline
	\multicolumn{2}{|l|}{Instance class} & \multicolumn{2}{c|}{Best known solution} & \multicolumn{5}{c|}{BonnTour default mode} & \multicolumn{5}{c|}{BonnTour high-effort mode}                                                                                                                                                                                                                                             \\
	\multicolumn{2}{|r|}{\#\,instances}  & tours                                    & distance                                   & tours                                          & \ (\%\,gap)              & distance             & \ (\%\,gap)              & $t$                  & tours           & \ (\%\,gap)              & distance             & \ (\%\,gap)              & $t$                                    \\
	\hline
	U (100--246)                         & 32                                       & 759                                        & \textbf{  1009991}                             & \REV{      777}          &                      & \REV{\textbf{  1019246}} & \REV{\textbf{(0.9)}} & \REV{      157} & \REV{      773}          &                      & \REV{\textbf{  1015816}} & \REV{\textbf{(0.6)}} & \REV{     1354} \\
	U (250--490)                         & 36                                       & 1567                                       & \textbf{  2040726}                             & \REV{     1594}          &                      & \REV{\textbf{  2070983}} & \REV{\textbf{(1.5)}} & \REV{      246} & \REV{     1587}          &                      & \REV{\textbf{  2059175}} & \REV{\textbf{(0.9)}} & \REV{     2051} \\
	U (501--1000)                        & 32                                       & 2683                                       & \textbf{  3259953}                             & \REV{     2755}          &                      & \REV{\textbf{  3316501}} & \REV{\textbf{(1.7)}} & \REV{      482} & \REV{     2744}          &                      & \REV{\textbf{  3296898}} & \REV{\textbf{(1.1)}} & \REV{     3292} \\
	GH (200)                             & 60                                       & \textbf{      694}                         & 168034                                         & \REV{\textbf{      694}} & \REV{\textbf{(0.0)}} & \REV{   173135}          &                      & \REV{      123} & \REV{\textbf{      694}} & \REV{\textbf{(0.0)}} & \REV{   170630}          &                      & \REV{     1004} \\
	GH (400)                             & 60                                       & \textbf{     1380}                         & 387654                                         & \REV{\textbf{     1386}} & \REV{\textbf{(0.4)}} & \REV{   404162}          &                      & \REV{      196} & \REV{\textbf{     1382}} & \REV{\textbf{(0.1)}} & \REV{   396773}          &                      & \REV{     1733} \\
	GH (600)                             & 60                                       & \textbf{     2065}                         & 784653                                         & \REV{\textbf{     2074}} & \REV{\textbf{(0.4)}} & \REV{   834728}          &                      & \REV{      331} & \REV{\textbf{     2068}} & \REV{\textbf{(0.1)}} & \REV{   811157}          &                      & \REV{     2714} \\
	GH (800)                             & 60                                       & \textbf{     2732}                         & \REV{1329290}                                  & \REV{\textbf{     2758}} & \REV{\textbf{(1.0)}} & \REV{  1392141}          &                      & \REV{      394} & \REV{\textbf{     2750}} & \REV{\textbf{(0.7)}} & \REV{  1360467}          &                      & \REV{     3402} \\
	GH (1000)                            & 60                                       & \textbf{     3416}                         & \REV{2019978}                                  & \REV{\textbf{     3435}} & \REV{\textbf{(0.6)}} & \REV{  2167893}          &                      & \REV{      494} & \REV{\textbf{     3423}} & \REV{\textbf{(0.2)}} & \REV{  2102217}          &                      & \REV{     3902} \\
	\REV{   LL (100)}                    & \REV{       56}                          & \REV{\textbf{      402}}                   & \REV{    58060}                                & \REV{\textbf{      402}} & \REV{\textbf{(0.0)}} & \REV{    58081}          &                      & \REV{       69} & \REV{\textbf{      402}} & \REV{\textbf{(0.0)}} & \REV{    58066}          &                      & \REV{      546} \\
	\REV{   LL (200)}                    & \REV{       60}                          & \REV{\textbf{      600}}                   & \REV{   183720}                                & \REV{\textbf{      608}} & \REV{\textbf{(1.3)}} & \REV{   183328}          &                      & \REV{      120} & \REV{\textbf{      602}} & \REV{\textbf{(0.3)}} & \REV{   184939}          &                      & \REV{      987} \\
	\REV{   LL (400)}                    & \REV{       60}                          & \REV{\textbf{     1129}}                   & \REV{   437440}                                & \REV{\textbf{     1160}} & \REV{\textbf{(2.7)}} & \REV{   451067}          &                      & \REV{      217} & \REV{\textbf{     1143}} & \REV{\textbf{(1.2)}} & \REV{   446268}          &                      & \REV{     1802} \\
	\REV{   LL (600)}                    & \REV{       60}                          & \REV{\textbf{     1619}}                   & \REV{   889271}                                & \REV{\textbf{     1673}} & \REV{\textbf{(3.3)}} & \REV{   909281}          &                      & \REV{      290} & \REV{\textbf{     1640}} & \REV{\textbf{(1.3)}} & \REV{   921799}          &                      & \REV{     2654} \\
	\REV{   LL (800)}                    & \REV{       60}                          & \REV{\textbf{     2102}}                   & \REV{  1483301}                                & \REV{\textbf{     2181}} & \REV{\textbf{(3.8)}} & \REV{  1522891}          &                      & \REV{      370} & \REV{\textbf{     2138}} & \REV{\textbf{(1.7)}} & \REV{  1524887}          &                      & \REV{     3567} \\
	\REV{  LL (1000)}                    & \REV{       58}                          & \REV{\textbf{     2562}}                   & \REV{  2184255}                                & \REV{\textbf{     2682}} & \REV{\textbf{(4.7)}} & \REV{  2244576}          &                      & \REV{      454} & \REV{\textbf{     2607}} & \REV{\textbf{(1.8)}} & \REV{  2259760}          &                      & \REV{     4517} \\
	\REV{       D (25)}                  & \REV{       56}                          &                                            & \textbf{\REV{   663546}}                       & \REV{      200}          &                      & \REV{\textbf{   664192}} & \REV{\textbf{(0.1)}} & \REV{       16} & \REV{      199}          &                      & \REV{\textbf{   664172}} & \REV{\textbf{(0.1)}} & \REV{      134} \\
	\REV{       D (50)}                  & \REV{       56}                          &                                            & \textbf{\REV{  1273212}}                       & \REV{      275}          &                      & \REV{\textbf{  1276041}} & \REV{\textbf{(0.2)}} & \REV{       43} & \REV{      274}          &                      & \REV{\textbf{  1274952}} & \REV{\textbf{(0.1)}} & \REV{      370} \\
	\REV{      D (100)}                  & \REV{       56}                          &                                            & \textbf{\REV{  2416331}}                       & \REV{      444}          &                      & \REV{\textbf{  2423347}} & \REV{\textbf{(0.3)}} & \REV{       99} & \REV{      429}          &                      & \REV{\textbf{  2420352}} & \REV{\textbf{(0.2)}} & \REV{      828} \\
	\hline
\end{tabular}

           }
          }
        \end{center}
        \caption{
                Relative gaps of the two modes of BonnTour compared to the best known solutions on popular benchmark instances for
                \textsc{CVRP} (U: Uchoa et al.~\cite{uchoa}, with a total of 100 instances ranging from 100 to 1000 customers),
                \textsc{CVRPTW} (GH: Gehring and Homberger~\cite{homberger}, with a total of 300 instances from 200 to 1000 customers),
                \textsc{PDPTW} (LL: Li and Lim~\cite{li2003metaheuristic}, with a total of 354 instances from 100 to 1000 tasks), and
                \textsc{TDVRPTW} (D: Dabia et al.~\cite{dabia2013branch}, with a total of 168 instances from 25 to 100 customers).
                The results for the main objective function is shown in bold. (For the \textsc{TDVRPTW} instances, the objective is to minimize the total working time, here shown as distance.) The average running time $t$ per instance is given in seconds.
                For the \textsc{TDVRPTW} instances, the number of tours are not reported in \cite{dabia2013branch, lera2018enhanced, LeraRomero2020, pan2021hybrid}.
                \label{tab:comparision_to_bks}
        }
\end{table}
We used two modes of BonnTour: the default mode that is also mostly used in practice and a high-effort mode with the
only difference that more random walk iterations are performed.
Table~\ref{tab:comparision_to_bks} shows results for different instance classes and sizes.
The two columns under ``Best known solution'' report the total number and total length of all tours of the best known solutions in that given class.
Then, we report our results both for the ``BonnTour default mode'' and the ``BonnTour high-effort mode''.
The results and the gaps of the primary objective to the best known solution are highlighted in bold.
The average running time $t$ per instance is given in seconds.

The results show that our algorithm produces good solutions in reasonable running time on these problems,
despite the overhead resulting from the more general implementation allowing time-dependent travel times.
The gap is consistently below 2\% in the high effort mode, and for instances with all pickups at a depot this
is even true in the much faster default mode.

Table~\ref{tab:comparision_to_bks} also contains results on the benchmarks for time-dependent capacitated vehicle routing with time windows by Dabia, R{\o}pke, Van Woensel and De Kok \cite{dabia2013branch}, which are generated from the Solomon instances \cite{solomon} by assigning three different types of speed profiles to the arcs. 
There are 168 instances ranging from from 25 to 100 customers. Solution values are reported in \cite{dabia2013branch,lera2018enhanced, LeraRomero2020, pan2021hybrid}. 
For nine instances with 100 customers each, we found a better solution than any of these papers\footnote{better solutions 
(instance/value): r112/17727.65, rc104/19256.52, c204/95608.05, r205/19314.75, r210/18520.02, r211/15389.66, rc205/23638.27, rc206/20607.01, rc208/15886.30}; 
the ``Best known solution'' value in Table~\ref{tab:comparision_to_bks} includes these new values.
The objective is to minimize the total working time, shown as distance in Table~\ref{tab:comparision_to_bks}.

\subsection{New benchmarks with time-dependent travel times}\label{sec:new-benchmarks}

To demonstrate the effectiveness of our algorithms and to foster future research on vehicle routing with time-dependent travel times,
we created a set of realistic instances with time-dependent travel times.
These are based on maps from OpenStreetMap (OSM, publicly available at \url{https://www.openstreetmap.org}) and speed
data that was published by Uber.
Our instances are available at \url{https://gitlab.com/muelleratorunibonnde/vrptdt-benchmark}.

We chose 10 cities from four sub-continents, for which we found a sufficient coverage of streets with Uber speed data.
We used data for the set $D$ of the ten Mondays from January 6 to March 9, 2020 (hence excluding effects of the Covid-19 pandemic on travel times).
For each road segment $r$ and each hour $h\in\{15,\ldots,21\}$,
let $D(r,h)\subseteq D$ be the set of these Mondays for which the data contains a speed profile for $r$ and $h$.
Then we defined speed values for $r$ and $h$ by computing
\begin{itemize}
  \item the average of the Uber speed data values provided for $r$ and $(d,h)$ for all days $d\in D(r,h)$, if $D(r,h) \not= \emptyset$;
  \item $p(h) \cdot f(r)$ otherwise, where $f(r)$ denotes the \emph{free-flow speed} on $r$ as defined below, and $p(h)$ the average fraction of free-flow speed during $[h, h+1)$, over all road segments fully covered by speed data within our time horizon.
\end{itemize}

Where possible we defined the free-flow speed $f(r)$ using the 85\,\% quantile speed values for $r$ provided by Uber in their Q1/2020 data set,
or alternatively using the OSM \texttt{maxspeed} value, or, as a last resort, based on the OSM \texttt{highway} attribute.

For each city we then chose a depot address, in most cases close to a major airport, 
and then defined three address sets with 2000, 1000, and 500 addresses.
We picked these addresses randomly from OSM nodes tagged as shops and not more than 500\,m away from a road segment covered by Uber speed data.
We computed (approximately) fastest time-dependent travel times for each
pair of start and end address as described in Section \ref{sec:preprocessing}.
The resulting matrix of arrival time functions is part of our published data.
We used the road network preprocessing described in Section~\ref{sec:preprocessing} to generate this data;  
now the road network is no longer needed to process our instances.

\begin{figure}[!t]
	\begin{subfigure}[b]{0.49\textwidth}
		\begin{center}
			\includegraphics[height=6.5cm]{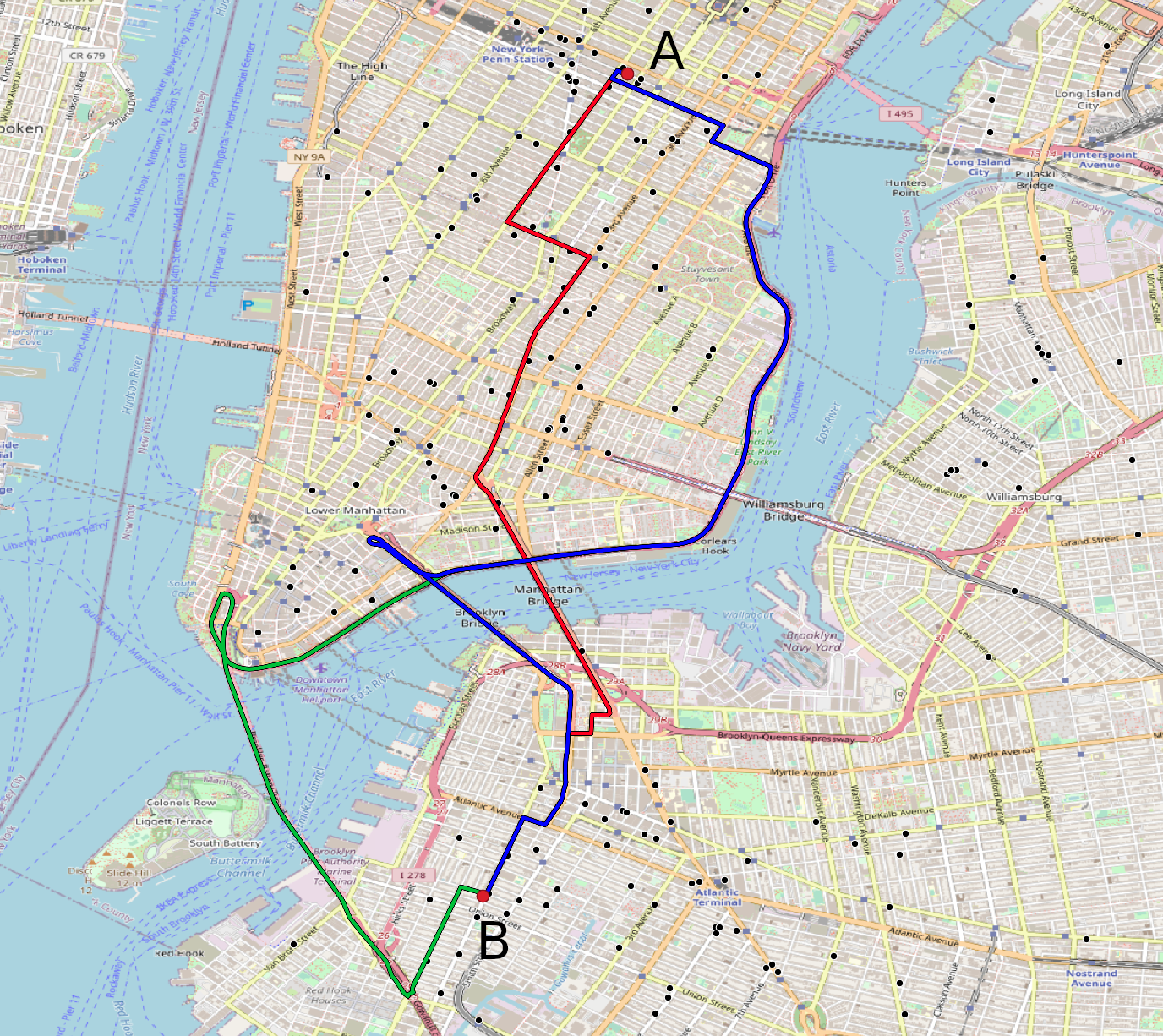}
		\end{center}
	\end{subfigure}
	\hfill \begin{subfigure}[b]{0.49\textwidth}
		\begin{center}
			\includegraphics[height=6.5cm]{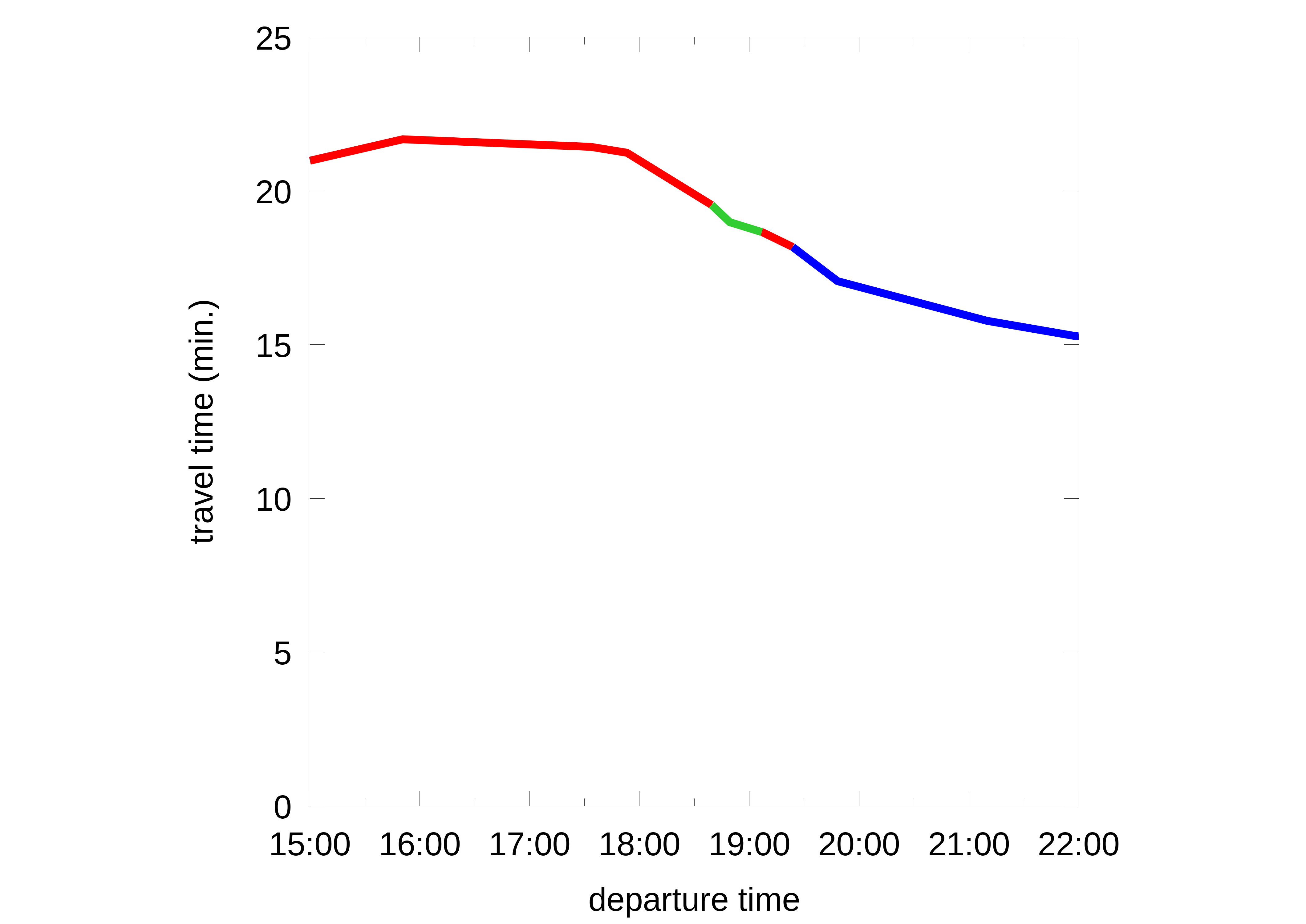}
		\end{center}
	\end{subfigure}
	\caption{The map on the left shows three paths in the 500-customer New York instance that are fastest at some point between 3\,pm and 10\,pm from some customer in Manhattan (A) to a customer in Brooklyn (B).
		On the right-hand side the corresponding travel time function is depicted. The segments are colored in the color of the path that is the fastest at this departure time.}
	\label{fig:alternative_fastest_paths}
\end{figure}

Figure~\ref{fig:alternative_fastest_paths} shows part of the addresses in the 500-customer New York instance.
Two addresses A and B are highlighted, and the paths from A to B that are the fastest path at some time between 15:00 and 22:00 are shown.
Figure~\ref{fig:alternative_fastest_paths} also shows the travel time function from A to B
(to obtain the arrival time function $a$ from the travel time function $f$, set $a(x):=x+f(x)$).

In all our instances there is only one type of vehicles. Each vehicle has a fixed cost of \$200 and unlimited capacity.
The objective is to minimize total fixed cost plus the working time cost, which we set to \$20 per hour.
The traveled distance is not part of our objective function. Nevertheless,
our published benchmark data also contains the length of each fastest path because it might be useful for future studies.

For each city and each address set we define two instances: one with all pickups at the depot, and one general pickup-and-delivery instance.
We first describe the instances with all pickups at the depot.
For each address $a$, we defined one item to be picked up at the depot at any time and delivered at $a$. Time windows for deliveries were chosen independently at random as follows. With 50\,\% probability we chose a one-hour time window beginning at a full or half hour chosen uniformly at random from
$\{\text{15:30},\text{16:00},\dots,\text{19:30},\text{20:00}\}$; the remaining items have a large delivery time window from 15:30 to 21:00.
We stipulate that each delivery takes three minutes and must start within the delivery time window, but
is allowed to extend past its end.
All tours must start at the depot no earlier than 15:00 and end at the depot at any time.
Pickup at the depot takes zero time, i.e., we assume pre-loaded vehicles.

Our general pickup-and-delivery instances are based on the same addresses and arrival time functions.
Here some of the addresses with a large time window are called \emph{shops}; the other addresses (except the depot) are called \emph{customers}. 
Now 80\,\% of the shipments are to be picked up at a shop and delivered to a customer, while 20\,\% of the shipments
are to be picked up at a customer and delivered to a shop. Details are described on the website containing the instances.

\subsection{Results on new benchmarks}
\label{sec:results-new-benchmarks}

\begin{table}[htb!]
        \begin{center}
                {\tiny
                        \resizebox{\linewidth}{!}{
                            \begin{tabular}{|l | r r r | r r r |}
\hline
Instance & \multicolumn{3}{c|}{BonnTour default mode} & \multicolumn{3}{c|}{BonnTour high-effort mode} \\
city \#\,customers & \#\,tours & cost & $t$ & \#\,tours & cost & $t$ \\
\hline\hline
        berlin\_500 & \REV{        9} & \REV{     2859} & \REV{      525} & \REV{        9} & \REV{     2827} & \REV{     3417} \\ 
    cincinnati\_500 & \REV{        9} & \REV{     2892} & \REV{      457} & \REV{        9} & \REV{     2907} & \REV{     2837} \\ 
          kyiv\_500 & \REV{       10} & \REV{     3170} & \REV{      576} & \REV{       10} & \REV{     3146} & \REV{     3592} \\ 
        london\_500 & \REV{       15} & \REV{     4847} & \REV{      424} & \REV{       15} & \REV{     4821} & \REV{     3569} \\ 
        madrid\_500 & \REV{       10} & \REV{     3202} & \REV{      507} & \REV{       10} & \REV{     3169} & \REV{     3562} \\ 
       nairobi\_500 & \REV{       10} & \REV{     3126} & \REV{      534} & \REV{       10} & \REV{     3113} & \REV{     3097} \\ 
     new\_york\_500 & \REV{       11} & \REV{     3432} & \REV{      504} & \REV{       11} & \REV{     3418} & \REV{     3434} \\ 
 san\_francisco\_500 & \REV{       12} & \REV{     3900} & \REV{      435} & \REV{       12} & \REV{     3854} & \REV{     2706} \\ 
    sao\_paulo\_500 & \REV{       13} & \REV{     4108} & \REV{      608} & \REV{       12} & \REV{     3910} & \REV{     3690} \\ 
       seattle\_500 & \REV{       10} & \REV{     3194} & \REV{      434} & \REV{       10} & \REV{     3179} & \REV{     3385} \\ 
\hline
 total & \REV{      109} & \REV{    34726} & \REV{     5004} & \REV{      108} & \REV{    34342} & \REV{    33289} \\ 
\hline\hline
        berlin\_1000 & \REV{       16} & \REV{     5027} & \REV{     1479} & \REV{       16} & \REV{     4981} & \REV{     6358} \\ 
    cincinnati\_1000 & \REV{       16} & \REV{     5114} & \REV{     1150} & \REV{       16} & \REV{     5059} & \REV{     6239} \\ 
          kyiv\_1000 & \REV{       18} & \REV{     5643} & \REV{     1281} & \REV{       17} & \REV{     5391} & \REV{     6770} \\ 
        london\_1000 & \REV{       25} & \REV{     8126} & \REV{      959} & \REV{       24} & \REV{     7871} & \REV{     6085} \\ 
        madrid\_1000 & \REV{       17} & \REV{     5437} & \REV{     1213} & \REV{       17} & \REV{     5421} & \REV{     6370} \\ 
       nairobi\_1000 & \REV{       17} & \REV{     5361} & \REV{     1368} & \REV{       17} & \REV{     5311} & \REV{     8297} \\ 
     new\_york\_1000 & \REV{       18} & \REV{     5700} & \REV{     1328} & \REV{       17} & \REV{     5455} & \REV{     6216} \\ 
 san\_francisco\_1000 & \REV{       20} & \REV{     6500} & \REV{      993} & \REV{       20} & \REV{     6456} & \REV{     6100} \\ 
    sao\_paulo\_1000 & \REV{       21} & \REV{     6802} & \REV{     1314} & \REV{       21} & \REV{     6753} & \REV{     8085} \\ 
       seattle\_1000 & \REV{       17} & \REV{     5397} & \REV{     1098} & \REV{       17} & \REV{     5389} & \REV{     6665} \\ 
\hline
 total & \REV{      185} & \REV{    59102} & \REV{    12183} & \REV{      182} & \REV{    58082} & \REV{    67186} \\ 
\hline\hline
        berlin\_2000 & \REV{       29} & \REV{     9077} & \REV{     3584} & \REV{       28} & \REV{     8850} & \REV{    13077} \\ 
    cincinnati\_2000 & \REV{       28} & \REV{     8955} & \REV{     2979} & \REV{       27} & \REV{     8723} & \REV{    12306} \\ 
          kyiv\_2000 & \REV{       31} & \REV{     9842} & \REV{     3379} & \REV{       30} & \REV{     9568} & \REV{    13559} \\ 
        london\_2000 & \REV{       42} & \REV{    13714} & \REV{     2501} & \REV{       40} & \REV{    13274} & \REV{    11882} \\ 
        madrid\_2000 & \REV{       30} & \REV{     9558} & \REV{     3009} & \REV{       29} & \REV{     9309} & \REV{    14270} \\ 
       nairobi\_2000 & \REV{       29} & \REV{     9110} & \REV{     3253} & \REV{       28} & \REV{     8888} & \REV{    16124} \\ 
     new\_york\_2000 & \REV{       31} & \REV{     9854} & \REV{     3220} & \REV{       30} & \REV{     9561} & \REV{    13851} \\ 
 san\_francisco\_2000 & \REV{       36} & \REV{    11533} & \REV{     2428} & \REV{       34} & \REV{    11022} & \REV{    12297} \\ 
    sao\_paulo\_2000 & \REV{       36} & \REV{    11493} & \REV{     3376} & \REV{       35} & \REV{    11259} & \REV{    16310} \\ 
       seattle\_2000 & \REV{       29} & \REV{     9283} & \REV{     2981} & \REV{       29} & \REV{     9180} & \REV{    13968} \\ 
\hline
 total & \REV{      321} & \REV{   102412} & \REV{    30711} & \REV{      310} & \REV{    99629} & \REV{   137643} \\ 
\hline
\end{tabular}
                         }
                }
        \end{center}
        \caption{Results of the two modes of BonnTour on our new benchmark instances with all pickups at a single depot. 
        In each row, the cost is rounded up to full dollars. The running time $t$ is given in seconds.
        \label{tab:results_new_benchmarks}}
\end{table}

\begin{table}[htb!]
   
        \begin{center}
                {\tiny
                        \resizebox{\linewidth}{!}{
                            \begin{tabular}{|l | r r r | r r r |}
\hline
Instance & \multicolumn{3}{c|}{BonnTour default mode} & \multicolumn{3}{c|}{BonnTour high-effort mode} \\
city \#\,customers & \#\,tours & cost & $t$ & \#\,tours & cost & $t$ \\
\hline\hline
        berlin\_500\_pdp & \REV{       15} & \REV{     4644} & \REV{      507} & \REV{       14} & \REV{     4404} & \REV{     3209} \\ 
    cincinnati\_500\_pdp & \REV{       15} & \REV{     4704} & \REV{      433} & \REV{       14} & \REV{     4495} & \REV{     2918} \\ 
          kyiv\_500\_pdp & \REV{       17} & \REV{     5327} & \REV{      470} & \REV{       16} & \REV{     4987} & \REV{     3044} \\ 
        london\_500\_pdp & \REV{       29} & \REV{     9471} & \REV{      498} & \REV{       28} & \REV{     9261} & \REV{     3196} \\ 
        madrid\_500\_pdp & \REV{       17} & \REV{     5311} & \REV{      636} & \REV{       16} & \REV{     5005} & \REV{     3857} \\ 
       nairobi\_500\_pdp & \REV{       17} & \REV{     5394} & \REV{      540} & \REV{       16} & \REV{     5046} & \REV{     3418} \\ 
     new\_york\_500\_pdp & \REV{       18} & \REV{     5702} & \REV{      556} & \REV{       17} & \REV{     5380} & \REV{     3647} \\ 
san\_francisco\_500\_pdp & \REV{       21} & \REV{     6670} & \REV{      500} & \REV{       19} & \REV{     6177} & \REV{     3182} \\ 
    sao\_paulo\_500\_pdp & \REV{       21} & \REV{     6650} & \REV{      658} & \REV{       20} & \REV{     6361} & \REV{     3981} \\ 
       seattle\_500\_pdp & \REV{       17} & \REV{     5418} & \REV{      492} & \REV{       16} & \REV{     5126} & \REV{     2920} \\ 
\hline
 total & \REV{      187} & \REV{    59288} & \REV{     5290} & \REV{      176} & \REV{    56237} & \REV{    33374} \\ 
\hline\hline
        berlin\_1000\_pdp & \REV{       27} & \REV{     8420} & \REV{     1264} & \REV{       26} & \REV{     8117} & \REV{     5758} \\ 
    cincinnati\_1000\_pdp & \REV{       27} & \REV{     8670} & \REV{      971} & \REV{       25} & \REV{     8048} & \REV{     4776} \\ 
          kyiv\_1000\_pdp & \REV{       34} & \REV{    10462} & \REV{     1049} & \REV{       30} & \REV{     9517} & \REV{     5590} \\ 
        london\_1000\_pdp & \REV{       54} & \REV{    17526} & \REV{     1146} & \REV{       51} & \REV{    16766} & \REV{     6082} \\ 
        madrid\_1000\_pdp & \REV{       32} & \REV{    10080} & \REV{     1333} & \REV{       30} & \REV{     9546} & \REV{     6194} \\ 
       nairobi\_1000\_pdp & \REV{       33} & \REV{    10367} & \REV{     1222} & \REV{       32} & \REV{    10067} & \REV{     6354} \\ 
     new\_york\_1000\_pdp & \REV{       33} & \REV{    10484} & \REV{     1260} & \REV{       32} & \REV{    10086} & \REV{     6051} \\ 
san\_francisco\_1000\_pdp & \REV{       38} & \REV{    12219} & \REV{     1157} & \REV{       35} & \REV{    11437} & \REV{     6226} \\ 
    sao\_paulo\_1000\_pdp & \REV{       41} & \REV{    13171} & \REV{     1349} & \REV{       38} & \REV{    12263} & \REV{     7467} \\ 
       seattle\_1000\_pdp & \REV{       30} & \REV{     9654} & \REV{     1008} & \REV{       29} & \REV{     9231} & \REV{     5987} \\ 
\hline
 total & \REV{      349} & \REV{   111049} & \REV{    11759} & \REV{      328} & \REV{   105074} & \REV{    60483} \\ 
\hline\hline
        berlin\_2000\_pdp & \REV{       54} & \REV{    17016} & \REV{     2887} & \REV{       51} & \REV{    16048} & \REV{    11711} \\ 
    cincinnati\_2000\_pdp & \REV{       54} & \REV{    17225} & \REV{     2569} & \REV{       50} & \REV{    15958} & \REV{    10441} \\ 
          kyiv\_2000\_pdp & \REV{       61} & \REV{    19362} & \REV{     3071} & \REV{       58} & \REV{    18398} & \REV{    10163} \\ 
        london\_2000\_pdp & \REV{      105} & \REV{    33530} & \REV{     2649} & \REV{       94} & \REV{    30870} & \REV{    13577} \\ 
        madrid\_2000\_pdp & \REV{       58} & \REV{    18548} & \REV{     2924} & \REV{       55} & \REV{    17506} & \REV{    11789} \\ 
       nairobi\_2000\_pdp & \REV{       63} & \REV{    19957} & \REV{     3033} & \REV{       60} & \REV{    19051} & \REV{    14007} \\ 
     new\_york\_2000\_pdp & \REV{       66} & \REV{    20880} & \REV{     3018} & \REV{       61} & \REV{    19410} & \REV{    14044} \\ 
san\_francisco\_2000\_pdp & \REV{       75} & \REV{    23985} & \REV{     2600} & \REV{       68} & \REV{    22048} & \REV{    12046} \\ 
    sao\_paulo\_2000\_pdp & \REV{       80} & \REV{    25401} & \REV{     3323} & \REV{       73} & \REV{    23651} & \REV{    14851} \\ 
       seattle\_2000\_pdp & \REV{       56} & \REV{    17813} & \REV{     2655} & \REV{       52} & \REV{    16642} & \REV{    10681} \\ 
\hline
 total & \REV{      672} & \REV{   213714} & \REV{    28730} & \REV{      622} & \REV{   199577} & \REV{   123310} \\ 
\hline
\end{tabular}
                         }
                }
        \end{center}
        \caption{Results of the two modes of BonnTour on our new pickup-and-delivery benchmark instances. 
        In each row, the cost is rounded up to full dollars. The running time $t$ is given in seconds.
        \label{tab:results_new_benchmarks_pd}}
   
\end{table}

In Tables~\ref{tab:results_new_benchmarks} and~\ref{tab:results_new_benchmarks_pd} we show the result of our algorithm
for the instances with 500, 1000, and 2000 customers in the 10 cities.
For each instance we show the number of tours and the total cost of our solution (computed by BonnTour default mode),
as well as the total running time of our algorithm.
We also show the corresponding data for the high-effort mode. 

As said in the beginning, most other vehicle routing algorithms still work with fixed travel times.
To assess the benefit of handling time-dependent travel times, we conducted several experiments where we tried to
compute good solutions when working only with fixed travel times.
To this end, we applied our algorithm in exactly the same way (in high-effort mode), but using constant travel times.
As travel times we took the worst case (in the domain of our ATFs), the average travel time, and the arithmetic mean of these two.
In the end we evaluated the resulting tours with respect to the time-dependent travel times.

\begin{table}[htb!]
        \begin{center}
                {\small
                        \resizebox{\linewidth}{!}{
                                \begin{tabular}{|l | r r@{} r| r r| }
 \multicolumn{6}{@{}l}{\REV{Results on the 30 new instances with all pickups at a single depot:}} \\
	\hline
	travel times                  & \#\,tours       & cost            & \ (\%\,gap)   & \#\,late deliveries & maximum delay (sec.) \\
	\hline
	time-dependent                & \REV{      600} & \REV{   192052} & \REV{ }       & 0                   & 0                    \\
	worst-case                    & \REV{      653} & \REV{   206877} & \REV{(7.7) }  & 0                   & 0                    \\
	50\,\% worst,\,50\,\% average & \REV{      629} & \REV{   199496} & \REV{(3.9) }  & \REV{      124}     & \REV{      197}      \\
	average                       & \REV{      600} & \REV{   191936} & \REV{(-0.1) } & \REV{     1656}     & \REV{      763}      \\
	\hline
 \multicolumn{6}{l}{} \\
 \multicolumn{6}{@{}l}{\REV{Results on the 30 new general pickup and delivery instances:}} \\
	\hline
	\REV{travel times}                  & \REV{\#\,tours} & \REV{cost}  & \ \REV{(\%\,gap)} & \REV{\#\,late pickups/deliveries} & \REV{maximum delay (sec.)} \\
	\hline
	\REV{time-dependent}                & \REV{     1126} & \REV{   360887} & \REV{ }       & \REV{        0}             & \REV{        0}      \\
	\REV{worst-case}                    & \REV{     1257} & \REV{   397555} & \REV{(10.2) } & \REV{        0}             & \REV{        0}      \\
	\REV{50\,\% worst,\,50\,\% average} & \REV{     1184} & \REV{   375862} & \REV{(4.1) }  & \REV{      406}             & \REV{      337}      \\
	\REV{average}                       & \REV{     1111} & \REV{   356550} & \REV{(-1.2) } & \REV{     3371}             & \REV{     1094}      \\
	\hline
\end{tabular}

                         }
                }
        \end{center}
        \caption{Results when planning with time-dependent and constant travel times. 
        Each row shows the total result summed over 30 new benchmark instances with 35000 tasks in total.
        For the single-depot instances (top), 17368 deliveries have a one-hour time window.
        For the pickup-and-delivery instances (bottom), 4638 pickups and 17051 deliveries have a one-hour time window.
        All results are evaluated with time-dependent travel times. 
        \label{tab:constant_tt}}
\end{table}

Table~\ref{tab:constant_tt} shows that working with worst-case travel times increases the cost by roughly 8--10\,\%.
Working with average travel times leads to similar total cost
but many unacceptable time window violations.
The arithmetic mean of worst-case and average travel times might be an acceptable compromise in practice,
but, as expected, the results are still much worse than using time-dependent travel times.

\begin{table}[htb!]
	\begin{minipage}{\linewidth}
        \begin{center}
                {\small
                        \resizebox{\linewidth}{!}{
                                \begin{tabular}{|l  l | r r@{} r | r r r r | }
 \multicolumn{6}{@{}l}{\REV{Results on the 30 new instances with all pickups at a single depot:}} \\
	\hline
	                     &                                 &                 &                 &               & \multicolumn{4}{c|}{\# deliveries with slack within:}                                                       \\
	travel times         & time windows                    & \#\,tours       & cost            & \ (\%\,gap)   & \REV{[10, 15)}                                        & \REV{[5, 10)}   & \REV{[0,5)}     & late            \\
	\hline
	time-dependent       & original                        & \REV{      600} & \REV{   192052} & \REV{ }       & \REV{     1691}                                       & \REV{     1721} & \REV{     2798} & 0               \\
	time-dependent       & soft reduction                  & \REV{      606} & \REV{   195038} & \REV{(1.6) }  & \REV{     1091}                                       & \REV{      792} & \REV{      263} & 0               \\
	\REV{time-dependent} & \REV{aggressive soft reduction} & \REV{      617} & \REV{   197856} & \REV{(3.0) }  & \REV{      691}                                       & \REV{      404} & \REV{       80} & 0               \\
	average              & original                        & \REV{      600} & \REV{   191936} & \REV{(-0.1) } & \REV{     1654}                                       & \REV{     1702} & \REV{     1860} & \REV{     1656} \\
	average              & reduced by 10 min               & \REV{      635} & \REV{   200744} & \REV{(4.5) }  & \REV{     2015}                                       & \REV{     1472} & \REV{      388} & \REV{       35$^*$} \\
	average              & soft reduction                  & \REV{      607} & \REV{   195292} & \REV{(1.7) }  & \REV{     1872}                                       & \REV{     1120} & \REV{      592} & \REV{      214} \\
	\REV{average}        & \REV{aggressive soft reduction} & \REV{      615} & \REV{   197606} & \REV{(2.9) }  & \REV{     1763}                                       & \REV{      873} & \REV{      373} & \REV{      157} \\
	\hline
 \multicolumn{6}{l}{} \\
 \multicolumn{6}{@{}l}{\REV{Results on the 30 new general pickup and delivery instances:}} \\
	\hline
	                     &                                 &                 &                 &               & \multicolumn{4}{c|}{\REV{\# pickups/deliveries with slack within:}}                                                       \\
	\REV{travel times}   & \REV{time windows}              & \REV{\#\,tours} & \REV{cost}  & \ \REV{(\%\,gap)} & \REV{[10, 15)}                                                & \REV{[5, 10)}   & \REV{[0,5)}     & \REV{late}      \\
	\hline
	\REV{time-dependent} & \REV{original}                  & \REV{     1126} & \REV{   360887} & \REV{ }       & \REV{     2326}                                               & \REV{     2561} & \REV{     4629} & \REV{        0} \\
	\REV{time-dependent} & \REV{soft reduction}            & \REV{     1131} & \REV{   364194} & \REV{(0.9) }  & \REV{     2035}                                               & \REV{     1818} & \REV{     1189} & \REV{        0} \\
	\REV{time-dependent} & \REV{aggressive soft reduction} & \REV{     1155} & \REV{   371201} & \REV{(2.9) }  & \REV{     1673}                                               & \REV{     1449} & \REV{      758} & \REV{        0} \\
	\REV{average}        & \REV{original}                  & \REV{     1111} & \REV{   356550} & \REV{(-1.2) } & \REV{     2381}                                               & \REV{     2538} & \REV{     2824} & \REV{     3371} \\
	\REV{average}        & \REV{reduced by 10 min}         & \REV{     1187} & \REV{   377126} & \REV{(4.5) }  & \REV{     2997}                                               & \REV{     2573} & \REV{     1235} & \REV{      604$^*$} \\
	\REV{average}        & \REV{soft reduction}            & \REV{     1131} & \REV{   363727} & \REV{(0.8) }  & \REV{     2881}                                               & \REV{     2134} & \REV{     1514} & \REV{     1154} \\
	\REV{average}        & \REV{aggressive soft reduction} & \REV{     1147} & \REV{   369559} & \REV{(2.4) }  & \REV{     3012}                                               & \REV{     2172} & \REV{     1378} & \REV{      878} \\
	\hline

\end{tabular}

                         }
                }                
        \end{center}
        \caption{Results when planning with the original hard time windows, soft time window reductions, or a hard reduction of every time window by 10 minutes. Each row shows the total result summed over 30 new benchmark instances with 35000 tasks in total. 
        For the single-depot instances (top), 17368 deliveries have a one-hour time window.
        For the pickup-and-delivery instances (bottom), 4638 pickups and 17051 deliveries have a one-hour time window.
        All results are evaluated with time-dependent travel times and the original hard time windows.
The costs do not include the artificial penalty for soft time window reductions.
The last four columns contain the number of deliveries with a certain slack until the deadline, e.g. $[10,15)$ contains the number of 
pickups/deliveries that happen at least 10 minutes but less than 15 minutes before the deadline.
{\boldmath $^*$}When planning with a hard reduction of every time window by 10 minutes, one customer in the single-depot S\~{a}o Paulo instances with 1000 and 2000 customers and one customer in the pickup-and-delivery London instance with 500 customers cannot be visited within its reduced delivery time window when starting at the depot at 3\,pm. 
Hence, this customer is not visited at all in the computed solution (and counted as late in the table), which is another drawback of this approach. 
\label{tab:stw}}
\end{minipage}
\end{table}

Another option to work with fixed travel times and still
meet most of the time windows, is to work with average travel times but reduce the length of the time windows upfront.
As Table~\ref{tab:stw} shows, a hard reduction of every time window by 10 minutes ends up with fewer time window violations,
but already leads to more than 4\,\% higher cost.

As described in Section~\ref{subsec:scheduling}, we can also handle soft time window reductions.
If we impose an artificial penalty of \$1 when arriving less than 15 but at least 10 minutes before the deadline,
\$2 within the next 5 minutes, and \$4 in the final 5 minutes 
and still strictly forbid arrivals after the deadline, the overall algorithm will avoid arrivals late in a time window
wherever this is possible at moderate cost. Table~\ref{tab:stw} shows the results: most arrivals at the end of the time window
can indeed be avoided at a moderate cost increase. 
If we double the artificial penalties (aggressive soft reduction in Table~\ref{tab:stw}), the effect is stronger as expected.
Soft time window reductions also help in combination with average travel times, but do not prevent late arrivals completely.
The combination of time-dependent travel times and soft time window reductions leads to schedules that are cheap and significantly more robust;
this is also illustrated by Figure~\ref{fig:slacks}.

\begin{figure}[htb!]
   \captionsetup[subfigure]{labelformat=empty}
	\begin{subfigure}[b]{0.49\textwidth}
		\begin{center}
			\includegraphics[height=5.5cm]{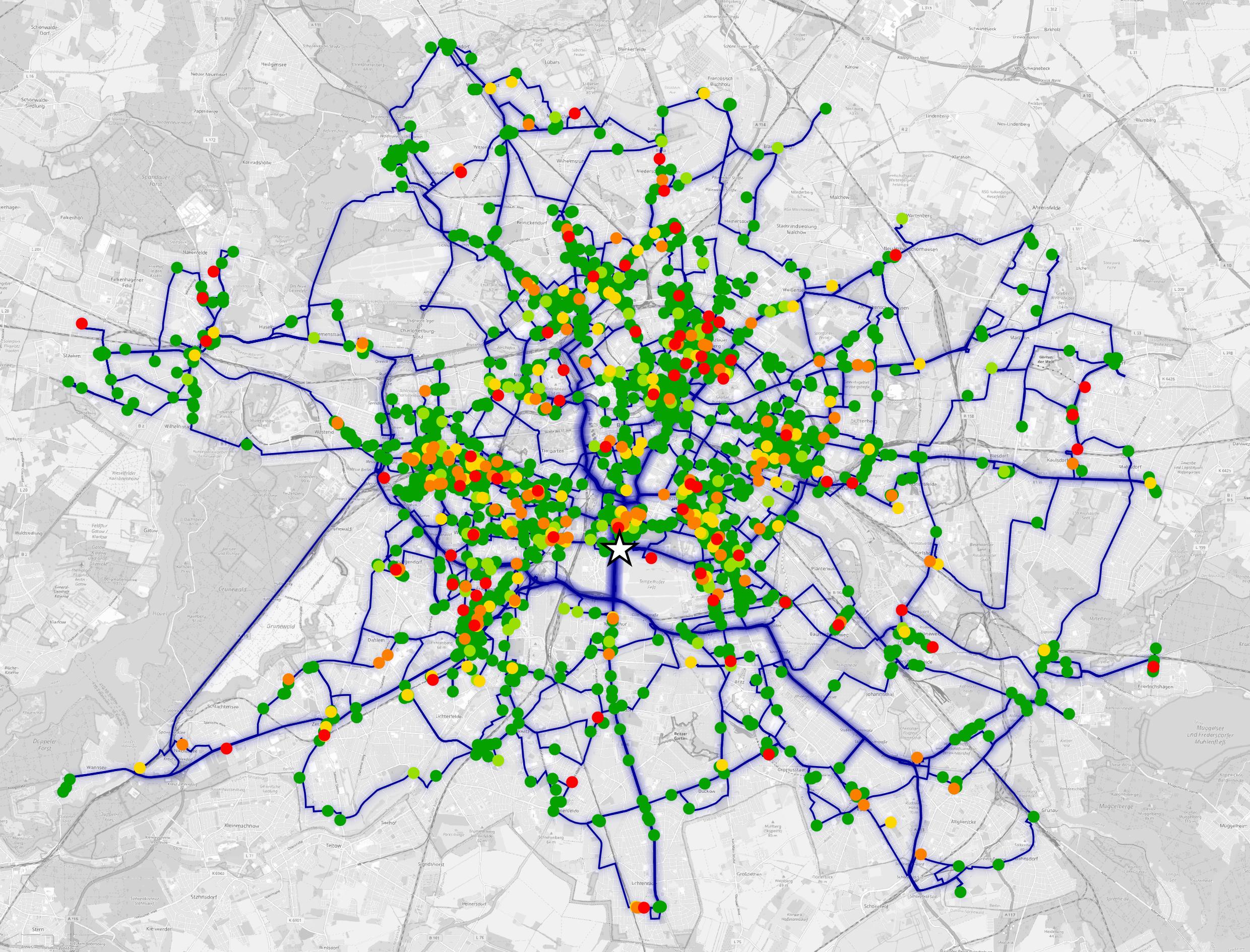}
		\end{center}
      \caption{{\bf a)} Constant average travel times: 28 tours, cost = 8859}
	\end{subfigure}
	\hfill \begin{subfigure}[b]{0.49\textwidth}
		\begin{center}
			\includegraphics[height=4cm]{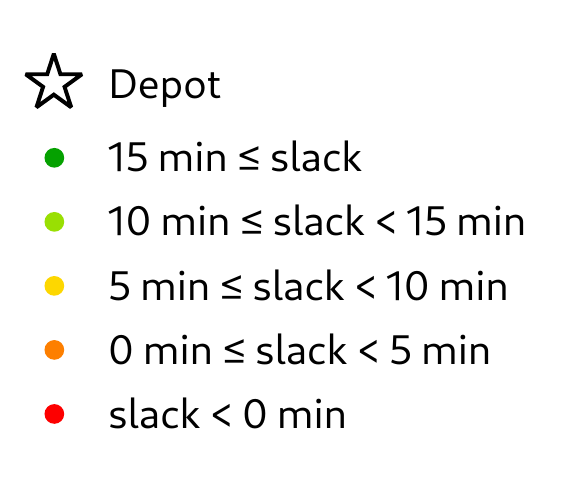}
		\end{center}
      \vspace{10mm}
	\end{subfigure}
	
	\vspace{0.75cm}
	\begin{subfigure}[b]{0.49\textwidth}
		\begin{center}
			\includegraphics[height=5.5cm]{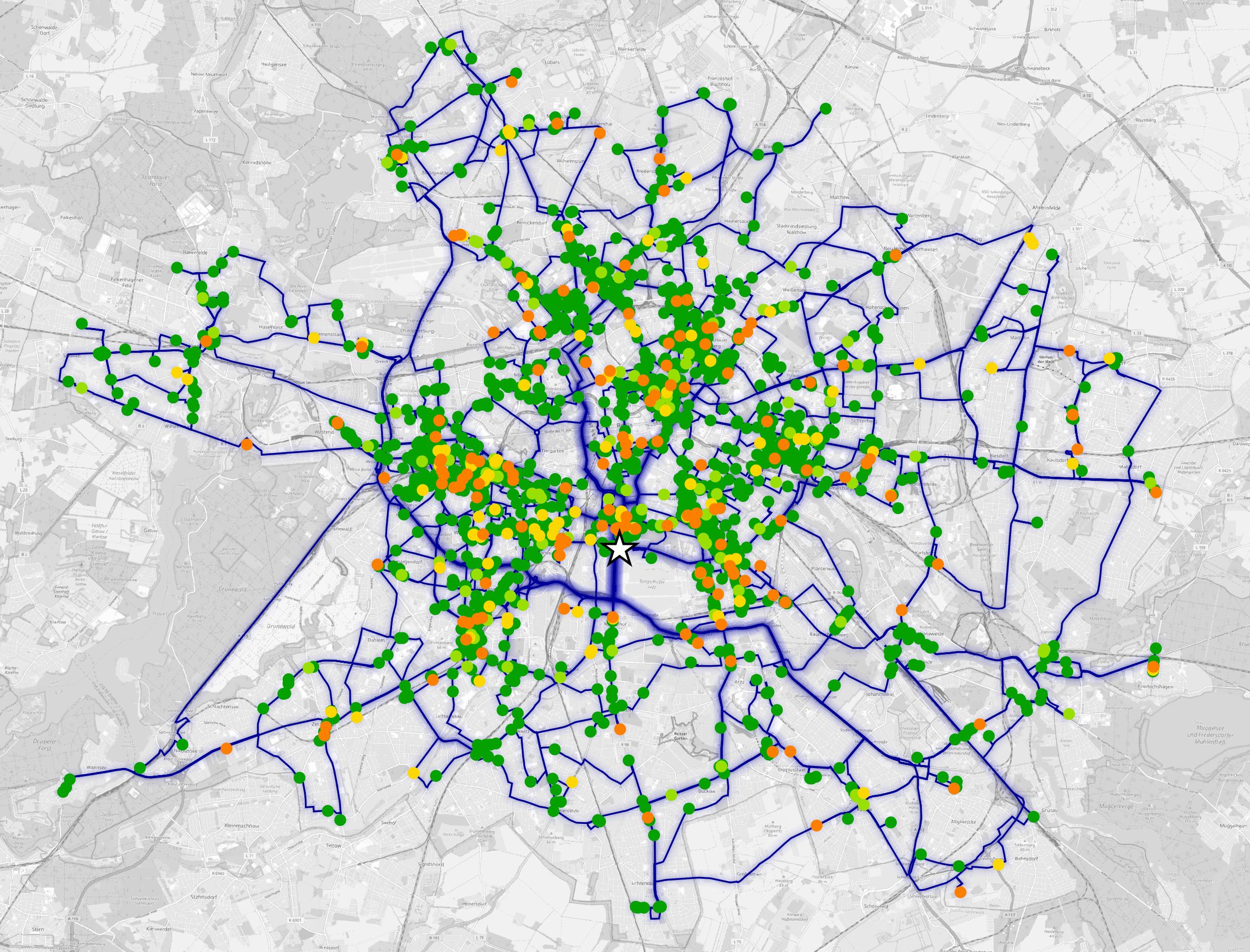}
		\end{center}
      \caption{{\bf b)} Time-dependent travel times: 28 tours, cost = 8850 \\ \mbox{}}
	\end{subfigure}
	\hfill \begin{subfigure}[b]{0.49\textwidth}
		\begin{center}
			\includegraphics[height=5.5cm]{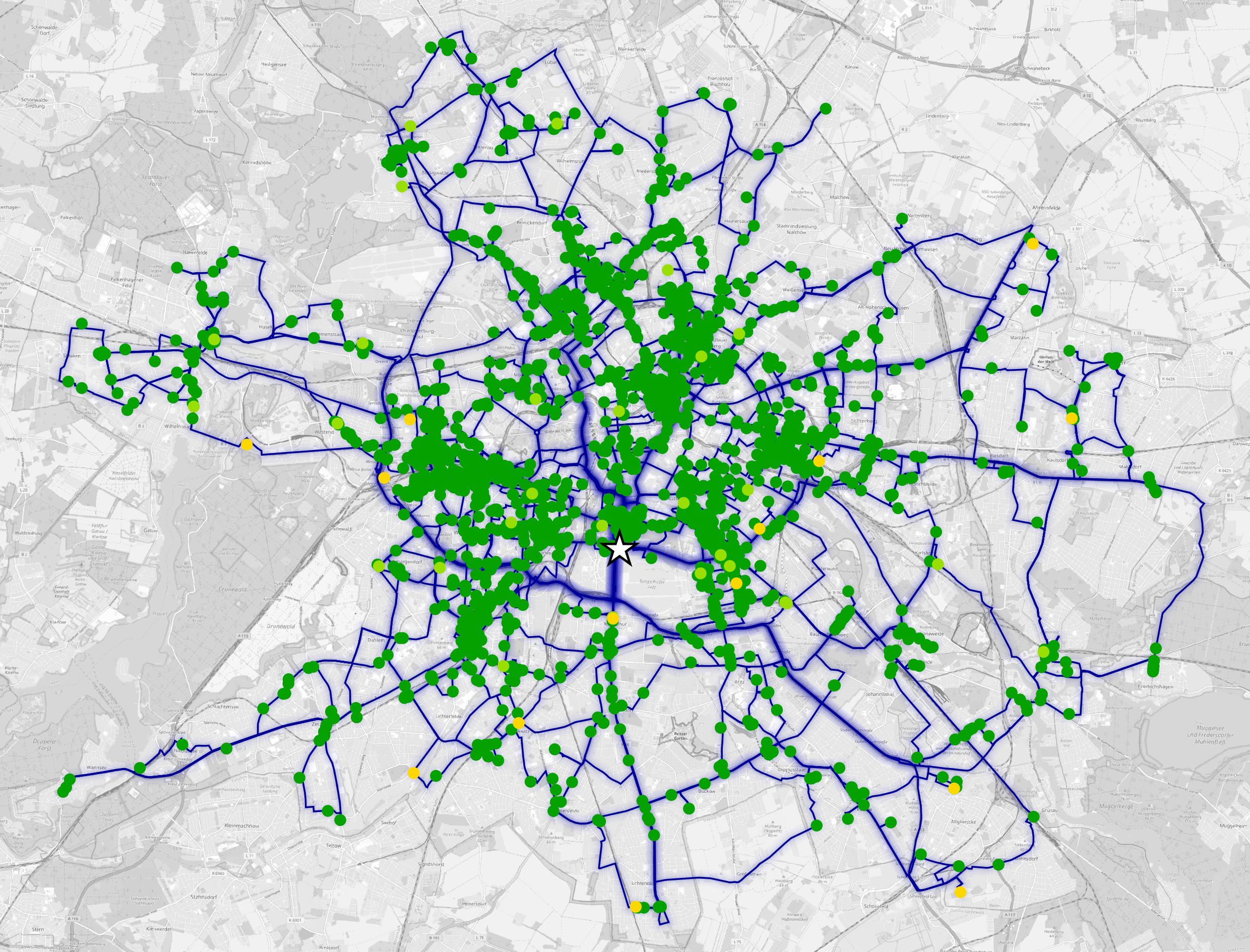}
		\end{center}
      \caption{{\bf c)} Time-dependent travel times and aggressive soft time window reduction: 29 tours, cost = 9149}
	\end{subfigure}

   \caption{Illustrating some results of Table~\ref{tab:stw}: three solutions obtained with high-effort mode on our Berlin benchmark instance with 2000 customers and all pickups at a single depot. Working with average travel times misses many time windows (red). Planning with time-dependent travel times and soft time window reductions largely avoids deliveries very late in the time window (yellow, orange).}
	\label{fig:slacks}
\end{figure}

\section{Conclusions}

We presented theoretical foundations and practical algorithms for vehicle routing with time-dependent travel times.
We have designed our algorithm for practical use, in cooperation with Greenplan, and it is already being used in practice in various scenarios.
We also constructed and published a set of realistic benchmark instances with time-dependent travel times and presented experimental results.
These demonstrate the benefit of working with time-dependent travel times.

\ifbool{journal}{
\clearpage
}{
\clearpage
\section*{Acknowledgement}
We thank all students who contributed to the implementation, in particular Luise Puhlmann and Silas Rathke.
We also thank our cooperation partner Greenplan, in particular Clemens Beckmann, Karin Pientka, and Jannik Silvanus.

We used map data copyrighted by OpenStreetMap contributors and available from \url{https://www.openstreetmap.org}, and
data retrieved from Uber Movement, \copyright\,2022 Uber Technologies, Inc., \url{https://movement.uber.com}.
}

 \bibliographystyle{acm}

\end{document}